\documentclass[reqno]{amsart}

\usepackage{amsmath,amsthm,amsfonts,amscd} 
\usepackage{amssymb}
\usepackage{enumitem}
\usepackage{mathrsfs}  	 
\usepackage{verbatim}      	 
\usepackage{url}		 

\newcommand{\D}{\Delta}
\newcommand{\pd}{\partial}
\newcommand{\g}{\gamma}

\renewcommand{\l}{\lambda}
\renewcommand{\L}{\Lambda}
\renewcommand{\r}{\rho}
\renewcommand{\t}{\tau}
\renewcommand{\a}{\alpha}
\newcommand{\n}{\nabla}
\newcommand{\w}{w}
\newcommand{\e}{\epsilon}
\newcommand{\la}{\langle}
\newcommand{\ra}{\rangle}
\newcommand{\s}{\sigma}
\newcommand{\sC}{\mathscr{C}}
\newcommand{\cC}{\mathcal{C}}
\newcommand{\sH}{\mathcal{H}}
\newcommand{\G}{\mathscr{G}}
\newcommand{\sN}{\mathcal{N}}
\renewcommand{\S}{\mathscr{S}}
\newcommand{\Ft}{\mathcal{F}_t}
\newcommand{\Pf}{\Psi_f}
\newcommand{\uu}{\underline{u}}
\newcommand{\uo}{u^{(0)}}
\newcommand{\uf}{u^{(1)}}
\newcommand{\us}{u^{(2)}}
\newcommand{\ut}{u^{(3)}}
\newcommand{\ufo}{u^{(4)}}

\newcommand{\uj}{u^{(j)}}

\newcommand{\vno}{v_{N_0}}
\newcommand{\uno}{u^{(0)}_{N_0}}
\newcommand{\unf}{u^{(1)}_{N_1}}
\newcommand{\uns}{u^{(2)}_{N_2}}
\newcommand{\unt}{u^{(3)}_{N_3}}
\newcommand{\unfo}{u^{(4)}_{N_4}}
\newcommand{\unj}{u^{(5)}_{N_5}}
\newcommand{\lu}{\underline{\lambda}}
\newcommand{\muu}{\underline{\mu}}

\newcommand{\vu}{\underline{v}}
\newcommand{\C}{\mathbb{C}}
\newcommand{\R}{\mathbb{R}}
\newcommand{\N}{\mathbb{N}}
\newcommand{\tor}{\mathbb{T}}
\newcommand{\Z}{\mathbb{Z}}
\newcommand{\cX}{\mathcal{X}}
\newcommand{\ov}{\overline}

\def\eqnn{\begin{eqnarray*}}
\def\eeqnn{\end{eqnarray*}}
\def\eqn{\begin{eqnarray}}
\def\eeqn{\end{eqnarray}}

\def\prf{\begin{proof}}
\def\endprf{\end{proof}}

\newtheorem{theorem}{Theorem}[section]
\newtheorem{cor}[theorem]{Corollary}
\newtheorem{lemma}[theorem]{Lemma}
\newtheorem{prop}[theorem]{Proposition}

\theoremstyle{definition}
\newtheorem{defin}{Definition}[section]

\theoremstyle{remark}
\newtheorem{rem}{Remark}[section] 

\numberwithin{equation}{section}

\begin{document}

\title[NLSS on the 3D Torus]{On the Well-posedness and Stability of Cubic and Quintic Nonlinear Schr\"odinger Systems on $\tor^3$}
                                                     
\author{Thomas Chen}  	 

\address[T.Chen]{Department of Mathematics, University of Texas at Austin} 

\email{tc@math.utexas.edu}

\author{Amie Bowles Urban}   

\address[A. Bowles Urban]{Department of Mathematics, University of Texas at Austin, and Johns Hopkins University Applied Physics Laboratory}

\maketitle

\begin{abstract}
In this paper, we study cubic and quintic nonlinear Schr\"odinger systems on 3-dimensional tori, with initial data in an adapted Hilbert space $H^s_{\lu},$ and all of our results hold on rational and irrational rectangular, flat tori. In the cubic and quintic case, we prove local well-posedness for both focusing and defocusing systems. We show that local solutions of the defocusing cubic system with initial data in $H^1_{\lu}$ can be extended for all time. Additionally, we prove that global well-posedness holds in the quintic system, focusing or defocusing, for initial data with sufficiently small $H^1_{\lu}$ norm. Finally, we use the energy-Casimir method to prove the existence and uniqueness, and nonlinear stability of a class of stationary states of the defocusing cubic and quintic nonlinear Schr\"odinger systems. 	
\end{abstract}

\section{Introduction}
\index{Introduction@\emph{Introduction}}%

In this work, we study properties of nonlinear Schr\"odinger systems on flat three-dimensional tori. Our results build on several lines of existing research: The study of nonlinear Schr\"odinger systems (NLSS) on $\R^d,$ the study of nonlinear Schr\"odinger equations (NLS) on flat tori, and the use of the energy-Casimir method to investigate certain stationary states of interacting quantum systems. 

The systems we consider may be used to model the dynamics of a system of fermions confined to a box with periodic boundary conditions. In particular, if we consider a dilute gas of fermions subject only to the pairwise interaction potential $\w,$ the one particle density operator of the system, $\g,$ solves the Landau-von Neumann equation with Hartree-type interaction,
\begin{equation*}
\begin{cases}
&i\pd_t\g = [-\D + \w\ast\r, \g] \\
&\g(t = 0)= \g_0
\end{cases}
\end{equation*}
where $\r$ is the total particle density given by $\r(t,x) = \g(t,x,x).$ This equation for $\g$ can be derived from the Schr\"odinger evolution equation for the wavefunction of the fermionic system through a combined mean-field and semiclassical limit, in which the expected particle number, $\text{Tr}\,\g,$ remains finite. See \cite{ElgErd} and \cite{FK} for details. 

If we take $\w$ to be a positive or negative delta function in the Hartree system above, and allow either two-particle or three-particle interactions, we obtain the system
\begin{equation}\label{Cauchy Sys}
\begin{cases}
&i\pd_t\g = [-\D \pm \r^{\a}, \g] \\
&\g(t = 0)= \g_0
\end{cases}
\end{equation}
The exponent $\a \in\{ 1,2\}$ indicates $(\a+1)$-body interactions, and the choice of sign on $\r^{\a}$ determines if the system is \emph{defocusing} $(+)$ or \emph{focusing} $(-).$
 
The one particle density operator $\g$ for a system of fermions is a positive, trace-class, self-adjoint operator on $L^2(\tor^3).$ Therefore, for each $t,$ its integral kernel $\g(t,x,y)$ has a spectral decomposition over $L^2(\tor^3).$ In particular, the initial data $\g_0(x,y)$  may be expressed as
\begin{equation}\label{DensOpInit}
\g(0,x,y) = \sum_{j \in \N}\l_j u_{j,0}(x)\overline{u_{j,0}}(y)
\end{equation} 
where $\{u_{j,0}\}_{j \in \N}$ is an orthonormal basis of $L^2(\tor^3),$ and $\lu:=\{\l_j \}_{j \in \N} \in \ell^1$ with $0 \leq \l_j \leq 1$ for all $j \in \N.$ 

Due to the commutator structure of \eqref{Cauchy Sys},  $\g$ and $i\D \mp i\s \r^{\a}$ form a Lax pair, hence the flow of $\g$ is isospectral, and $\{\l_j \}_{j \in \N}$ is constant in time. The evolution of $\g$ is therefore given by the evolution of the functions $\uu:=\{u_{j}\}_{j \in \N},$ and we may write
\begin{equation}\label{eq-gamma-spec-1}
\g(t,x,y) = \sum_{j=1}^{\infty} \l_j  u_{j}(t,x)\overline{u_{j}}(t,y),
\end{equation}
where the set $\{u_j\}_{j \in \N}$ remains orthonormal as long as the solution $\g$ exists. The particle density is given by $\r(t,x)\equiv\r_{\uu(t),\lu } = \g(t,x,x)$ so that in terms of the basis $\{u_{j}\}_{j \in \N},$
\begin{equation*}
\r(t,x) =  \sum_{j=1}^{\infty} \l_j  |u_{j}(t,x)|^2.
\end{equation*}
The Landau-von Neumann equations for $\g(t)$ in equation \eqref{Cauchy Sys} then have the form
\begin{eqnarray*}
	\lefteqn{
	i\partial\gamma(t,x,y)=
	\sum_{j=1}^{\infty} \l_j  ((i\partial_t u_{j})(t,x)\overline{u_{j}}(t,y)
	- u_{j}(t,x)\overline{((i\partial_t u_{j})}(t,y))
	}
	\nonumber\\
	&=&
	\sum_{j=1}^{\infty} \l_j \Big[((-\Delta+\sigma\rho^\alpha) u_{j})(t,x)\overline{u_{j}}(t,y) -  u_{j}(t,x)(\overline{(-\Delta+\sigma\rho^\alpha) u_{j}})(t,y)\Big]
\end{eqnarray*}
in the spectral decomposition \eqref{eq-gamma-spec-1}.
This is equivalent to the infinite nonlinear Schr\"odinger system (NLSS) for $\uu(t)=\{u_{j}(t)\}_{j \in \N}$, 
\begin{equation}\label{NLSS}
\begin{cases}
& i\pd_t u_j = -\D u_j+ \s\r^{\a} u_j, \qquad j \in \mathbb{N}\\
&u_j(0,x) = u_{j,0}(x), \qquad x \in \tor^3,
\end{cases}
\end{equation}
with  $\a \in\{ 1,2\}$ and $\s \in \{-1,1\}$.
The initial data $\{u_{j,0}\}$ for the NLSS and the sequence $\{\l_j \}$ are determined by the initial data $\g(0)$ in the Cauchy problem \eqref{DensOpInit}.

In this paper, we extend previous results of Markowich, Rein, and Wolanski,  \cite{MRW}, and of Abou Salem, Chen, and Vougalter, \cite{AsCV}, proving the existence and nonlinear stability of a class of stationary states of Schr\"odinger-Poisson systems via the energy-Casimir method. This approach is based on the fact that the sequence $\lu=\{\lambda_j\}$ is conserved under the NLSS flow, and uses it to construct an energy-Casimir functional $\sH_f$, labeled by a Casimir class function $f$, see Definition \ref{def-Casimir-fct-1}. $\sH_f$ then is a conserved quantity of the NLSS flow for any such $f$. The stability of stationary solutions of the NLSS is proven by use of $\sH_f$ in a similar way as Lyapunov functions are used for the corresponding problem in classical Hamiltonian dynamics. 
In particular, the stationary states arise as minimizers of energy-Casimir functionals, which are conserved quantities of the system. 
To be more precise, let $(\uu_0,\lu_0, \r_0)$ label a stationary state of the defocusing NLSS with $\sigma=1$ and $\a \in \{1,2\}$, and let $(\uu(t), \lu)$ account for another solution on the time interval $[0,T)$ with $T\leq\infty$, and initial datum $(\uu(0), \lu)$ (see Theorem \ref{StabThm} for the precise formulation), then 
\begin{equation*}
\frac{1}{\a+1} \| \r_{\uu(t), \lu} - \r_0\|^{\a+1}_{L^{\a+1}(\tor^3)} \leq 
|\sH_f(\uu(0),\lu) - \sH_f(\uu_0,\lu_0)|,  
\end{equation*}
for all $t \in [0, T).$

In both \cite{MRW} and \cite{AsCV}, classical solutions to the system were considered; hence, higher regularity was required than that controlled by the conserved energy. The nonlinear stability is obtained from a uniform in time upper bound on the squared distance (measured in some Sobolev norm) between $\r_0$ and $\r,$ where $\r_0$ is the particle density for a stationary state, and $\r$ is the particle density of another solution of the system. 

The energy-Casimir method employed in \cite{MRW} and \cite{AsCV} requires that the system is posed on a bounded spatial domain, that the flow of the system is isospectral, and that the potential function of the Hamiltonian is related to the probability density function. As the last two properties hold for the NLSS, it is natural to consider whether the NLSS possesses such stationary states. For this purpose, we pose the system on a bounded spatial domain, or more specifically, on $\tor^3$. In particular, we relax the criteria on the regularity of stationary states, using only mild solutions in the energy space, for which we establish well-posedness. 

Systems similar to \eqref{Cauchy Sys} have been previously studied on $\R^d$, and well-posedness results have been obtained under various assumptions on $\g_0.$ In particular, Hong, Kwon, and Yoon \cite{Hong} established the well-posedness theory and blow-up criteria for \eqref{Cauchy Sys} with $\a=1$ on $\R^3$ for $\g_0$ satisfying $\text{Tr}|\sqrt{-\D}\g_0\sqrt{-\D}| < \infty.$ In \cite{ChenPav}, Chen, Hong and Pavlovi\'{c} proved the global well-posedness of the defocusing system with $\a =1$ on $\R^2$ and $\R^3$ in the case $\g_0$ is not trace-class, provided it has finite operator norm and is a suitable perturbation of a reference state. 

In the present work, we will employ methods and results from the study of well-posedness for NLS on $\tor^d$. Following a series of fundamental works by Bourgain starting in 1993, \cite{bourg1}, this topic has attracted extensive research activity. Crucial advances include the development of Strichartz estimates on the torus and their extensions to irrational square tori, due to works by Bourgain \cite{bourg3}, Bourgain and Demeter \cite{BourDem}, and Guo, Oh, and Wang \cite{GOW}, among many others; Killip and Vi\c{s}an  proved the full range of Strichartz estimates on rational and irrational rectangular tori in \cite{kv}. We refer to those works for references.

Our analysis of the quintic NLSS is closely related to that of the $H^1$-critical quintic NLS on $\tor^3.$ Of specific importance for our work is the approach developed by Herr, Tataru, and Tzvetkov via $X^s$ and $Y^s$ function spaces, used in \cite{HTaTzv} to prove local and global well-posedness for the quintic NLS with small initial data in $H^1(\tor^3).$ Killip and Vi\c{s}an extended these results in \cite{kv}, proving local well-posedness of the $H^1$-critical NLS on rational and irrational rectangular tori in $3$ and $4$ dimensions for arbitrary initial data in $H^1.$ In \cite{IonPau}, Ionescu and Pausader obtained global well-posedness of the defocusing quintic NLS on the square torus for arbitrary initial data in $H^1(\tor^3).$ We show local well-posedness on $\tor^3$ for the cubic NLSS with initial data in $H^s_{\lu}$ with $s>\frac{1}{2}$. 
and for the quintic NLSS with initial data in $H^1_{\lu}.$ Furthermore, we prove that solutions to the defocusing cubic NLSS can be extended globally in time, as can solutions to the quintic NLSS with sufficiently small initial data.

We now outline our results and the organization of this paper.  In Sections \ref{sec-2} to \ref{sec-5}, we prove local well-posedness of the cubic NLSS, \eqref{Cauchy Sys} with $\a=1,$ on a flat rational or irrational 3-torus for initial data in $H^s_{\lu},$ for $s>\frac{1}{2}$. 
In Sections \ref{sec-6} to \ref{sec-9}, we prove the local well-posedness in $H^1_{\lu}$ for the quintic NLSS, \eqref{Cauchy Sys} with $\a =2$, using 
the $X^s$ and $Y^s$ spaces as in \cite{HTaTzv} and \cite{kv}.  In Sections \ref{sec-10} to \ref{sec-14}, we define a class of stationary states for the NLSS on $\tor^3$ corresponding to a Casimir function $f,$ treating both the cubic and quintic systems. Assuming their existence, we first prove the nonlinear stability of these stationary states using an energy-Casimir functional. We then use a dual formulation and tools from convex analysis to prove the existence and uniqueness of the stationary states and show that they are indeed minimizers of the energy-Casimir functional determined by $f.$
\section{Preliminaries}
\label{sec-2}
The rectangular, flat 3-torus can be realized as $\R^3/(L_1\Z \times L_2\Z \times L_3\Z)$ with $L_1, L_2, L_3 \in (0, \infty).$ The torus is \emph{irrational} if at least one of the ratios $\frac{L_i}{L_j}$ is irrational, otherwise we say it is \emph{rational.}

 For notational convenience, we use the coordinates for the standard torus $\tor^3:=\R^3 / \Z^3$ and incorporate the geometry of the torus into the Riemannian metric, using the corresponding Laplace-Beltrami operator 
\begin{equation*}
\D = \theta_1 \frac{\pd^2}{\pd x_1^2} +  \theta_2 \frac{\pd^2}{\pd x_2^2} +  \theta_3 \frac{\pd^2}{\pd x_3^2},\quad\text{where }\theta_j = L_j^{-2}.
\end{equation*}
We then define the Schr\"odinger propagator $e^{it\D}$ by
\begin{equation*}
\widehat{ e^{it\D}f }(\xi) = \exp\left(-2\pi i t Q(\xi)\right)\hat{f}(\xi)
\end{equation*}
for $\xi = (\xi_1, \xi_2, \xi_3) \in \Z^3,$ where $Q(\xi) := \theta_1\xi_1^2 + \theta_2\xi_2^2 + \theta_3 \xi_3^2.$ By making a change of variables in time, we may assume $\theta_j \in (0,1],$ for each $j\in\{1,2,3\}.$

Next, we define the Littlewood-Paley frequency projections used in Chapters 2 and 3. Let $\phi$ be smooth, radial, cutoff on $\R$ with $supp(\phi) \subset (-2,2)$ such that $\phi(x) = 1$ for $x \in [-1,1].$ For a dyadic integer $N,$ define the projections
\begin{align*}
\widehat{P_1 f}(\xi) &:= \hat{f}(\xi)\prod_{j=1}^3 \phi(\xi_j)\\
\widehat{P_{ \leq N} f}(\xi) &:= \hat{f}(\xi)\prod_{j=1}^3 \phi(\tfrac{\xi_j}{N})\\
\widehat{P_N f}(\xi) &:= \hat{f}(\xi)\prod_{j=1}^3\left( \phi(\tfrac{\xi_j}{N}) - \phi(\tfrac{2\xi_j}{N})\right).
\end{align*}
For  $\cC_N \subset \R^3$ an arbitrary cube of side length $N,$ the sharp Fourier projection onto $\cC_N$ is given by
\begin{equation*}
\widehat{P_{\cC_N} f}(\xi) = \mathbf{1}_{\cC_N}(\xi)\hat{f}(\xi).
\end{equation*}

We close this chapter with the following overview of the notational conventions we use in this work.
\begin{enumerate}[label=\textbullet]
\item We write $X \lesssim Y$ to represent $X \leq CY$ where $C$ is some constant that is permitted to depend only on the spatial dimension $d.$
\item Unless otherwise indicated, the domain of a spatial integral is understood to be $\tor^3,$ i.e.
\[\int f(x)\,dx:=\int_{\tor^3}f(x)\,dx\]
\item An underlined variable denotes a sequence in the corresponding variable, e.g. $\vu:=\{v_j\}_{j \in \N}.$
\item For any set $\cX,$ with elements that are real-valued, $\cX_+$ denotes the subset
\[\cX_+:=\{f \in \cX \big\arrowvert f \geq 0\}\]
\item We adopt the following condensed notation for frequency projections:
\begin{equation*}
f_N:=P_N f \qquad\text{and}\qquad f_{\cC_N}:=P_{\cC_N}f
\end{equation*}
\item We use the mixed space-time norms defined by
\[ \|f(t,x)\|_{L^p_t L^q_x([0,T) \times \tor^3)}:=\left(\int_0^T\left(\int_{\tor^3}|f(t,x)|^q\,dx\right)^{\frac{p}{q}}dt\right)^{\frac{1}{p}}\]
\item Given a Banach space $\cX$ and a real-valued sequence $\lu \in \ell^1_+,$ let $\cX_{\lu}$ denote the space of sequences $\uu = \{u_j\}_{j=1}^{\infty} \subset \cX$  equipped with the norm
\begin{equation*}
\| \uu \|_{\cX_{\lu}}:=\left(\sum_{j=1}^{\infty} \l_j \| u_j \|^2_{\cX}\right)^{\frac{1}{2} }.
\end{equation*}
\end{enumerate}

\section{Well-posedness of the Cubic NLS on $\tor^3$}
\index{Well-posedness of the Cubic NLSS on $\tor^3$%
@\emph{Well-posedness of the Cubic NLSS on $\tor^3$}}%

\label{sec-2}

The cubic nonlinear Schr\"odinger system is given by
\begin{equation}\label{CNLSS}
\begin{cases}
& i\pd_t u_j =-\D u_j+ \s\r u_j, \qquad j \in \mathbb{N}\\
&u_j(0,x) = u_{j,0}(x), \qquad x \in \tor^3,
\end{cases}
\end{equation}
where $\s \in \{-1,1\},$ $\lu \in \ell^1_+$, and $\r(t,x) = \sum \l_j |u_j(t,x)|^2.$ The mass and energy, 
\begin{align}\label{CubCons}
M_{\lu}(\uu)&:=\sum_k \l_k \|u_k\|_{L^2(\tor^3)}^2 = \|\uu\|_{L^2_{\lu}}^2 \\
E_{\lu}(\uu) &:=\frac{1}{2}\sum_k \l_k \|\n u_k\|_{L^2(\tor^3)}^2 +\s \frac{1}{4}\int_{\tor^3} \r^2\,dx 
\end{align}
are conserved quantities along solutions of the system. This chapter is dedicated to the proof of the following theorem:
\begin{theorem}[Local and global well-posedness of the cubic NLSS]\label{CubMainThm}
Let  $\lu \in \ell^1_+,$ and suppose $\uu_0 \in H_{\lu}^s(\tor^3)$ for $s>\frac{1}{2}$. 
There exists a time $T$ depending on $\|\uu_0\|_{H_{\lu}^s(\tor^3)}$ such that the system \eqref{CNLSS} is locally well-posed for $t \in [0,T)$. Moreover, if $\uu_0 \in H_{\lu}^1(\tor^3),$ the solution to the defocusing system is global in time.
\end{theorem}


Our goal is to use the contraction mapping principle to show that the Duhamel formula corresponding to \eqref{CNLSS} has a fixed point. In order to bound the terms of the Duhamel formula in the desired function space, we will decompose factors of the nonlinear function $|u_k|^2u_j$ frequency cubes, apply the appropriate Strichartz estimates on each frequency cube, and find an upper bound for the sum over all such decompositions. Thus, the primary tools we use are the following Strichartz estimates on $\tor^d,$ due to Killip and Vi\c{s}an:
\begin{theorem}{\cite{kv}}\label{Strich1}
For $d\geq 1$, $\theta_1,...\theta_d \in (0,1],$ $1\leq N \in 2^{\Z},$ and $p > \frac{2(d+2)}{d}.$ Then,
\begin{equation}\label{KVStrich}
\| e^{it\D}P_{\leq N} f\|_{L^p_{t,x}([0,1]\times \tor^d)} \lesssim N^{\frac{d}{2} - \frac{d+2}{p}} \| f \|_{L^2(\tor^d)}
\end{equation}
where $\D:=\theta_1 \pd_{x_1}^2 + ... + \theta_d \pd_{x_d}^2$
\end{theorem}

As we only consider problems posed on $\tor^3,$ we note that the above inequality with $d =3$ reads
\begin{equation*}
\| e^{it\D}P_{\leq N} f\|_{L^p_{t,x}([0,1]\times \tor^3)} \lesssim N^{\frac{3}{2} - \frac{5}{p}} \| f \|_{L^2(\tor^3)}
\end{equation*}
for $p > \frac{10}{3}.$

\begin{rem}
Due to the invariance of $e^{it\D}f(x)$ under Galileian transformations,  if $\cC_N$ is a cube of side length $N$ in $\R^3$ and $p > \frac{10}{3},$ we have 
\begin{equation*}
\|e^{it\D}P_{\cC_N} f\|_{L^p_{t,x}([0,1]\times \tor^3)} \lesssim N^{\frac{3}{2} - \frac{5}{p}} \| f \|_{L^2(\tor^3)}
\end{equation*}
\end{rem}


The Bourgain space $X^{s,b}:=X^{s,b}(\R \times \tor^3)$ is the completion of $C^{\infty}\big(\R; H^s(\tor^3)\big)$ under the norm
\begin{align*}
\|u\|_{X^{s,b}} :&= \|e^{-it\D}u(t,x)\|_{H^b_t(\R;H^s_x(\tor^3))}\\ &=\Big( \sum_{\xi \in \Z^3 }\int _\R d\tau \la \tau + Q(\xi)\ra^{2b}\la \xi \ra^{2s}|\widehat{u}(\tau,\xi)|^2\,d\tau\Big)^{\frac{1}{2}},
\end{align*}
where $Q(\xi) := \theta_1\xi_1^2 + \theta_2\xi_2^2+\theta_3\xi_3^2.$ For $0 < T \leq 1,$ define the restriction space $X^{s,b}_T := X^{s,b}([0,T] \times \tor^3)$ with the norm 
\begin{equation*}
\|u \|_{X^{s,b}_T} = \inf_{w \in X^{s,b}}\left\{ \|w\|_{X^{s,b}},\text{ with } w\arrowvert_{[0,T]} = u \right\}
\end{equation*}

\begin{rem}\label{EmbRem}
We will make use of the following embedding properties of the $X^{s,b}$ spaces:
\begin{enumerate}[label = (\arabic*), labelindent=\parindent]
\item For $s_1 \leq s_2$ and $b_1 \leq b_2,$ $X^{s_2,b_2} \hookrightarrow X^{s_1, b_1}.$
\item For $b > \frac{1}{2},$ $X^{0,b} \hookrightarrow C_t L^2_x.$
\item $X^{0,\frac{1}{4}} \hookrightarrow  L^4_tL^2_x.$
\end{enumerate}
Property (1) is a direct consequence of the definition of the $X^{s,b}$ norm and monotonicity. Property (2) follows from the observation that $\la \tau + Q(\xi) \ra^{-b} \in L^2_\tau(\R)$ for $b >\frac{1}{2}.$ Property (3) can be shown by the Sobolev embedding $H^{\frac{1}{4}}(\R) \hookrightarrow L^4(\R)$ applied to the $ L^4_tL^2_x$ norm of $e^{it\D}U(t,x)$ for $U(t,x) = e^{-it\D}u(x,t).$
\end{rem}


\section{Nonlinear Estimates for the Cubic NLSS}

The following proposition, due to Ginibre, gives an upper bound for $X^{s,b}_T$ norm of the nonlinear term of the Duhamel formula, thus it motivates the development of the nonlinear estimates in this section. We refer the interested reader to \cite{BGTzv} and \cite{Gin} for the proof of the proposition.
\begin{prop}\label{DuhamBds}
Suppose $0 < T \leq 1.$ For $(b,b') \in \R^2$ satisfying $0 < b' < \frac{1}{2} < b$ and $b + b' < 1,$
\begin{equation*}
\left\Vert \int_0^t e^{i\D(t - t')}F(t')\,dt\right\Vert_{X^{s,b}_T} \lesssim T^{1 - b - b'}\|F \|_{X^{s,-b'}_T}
\end{equation*}
\end{prop}

The next lemma is the crucial nonlinear estimate for local well-posedness of the cubic NLSS in $H^s(\tor^3)$ for $s>\frac{1}{2}$, 
which we will prove in this section.
\begin{lemma}\label{ThreeFuncs}
Let $s>\frac{1}{2}$. 
 There exists $C > 0$ and $(b, b') \in \R^2$ with $\frac{1}{4} < b'<\frac{1}{2} < b$ satisfying $b + b' < 1,$ such that for every triple $(\uf, \us, \ut)$ with $\uj \in X^{s,b}(\R \times \tor^3)$ for $j = 1, 2, 3,$
\begin{equation}
\|\uf \ov{\us} \ut \|_{X^{s,-b'}(\R \times \tor^3)} \leq C \prod_{j=1}^{3}\|u^{(j)}\|_{X^{s,b}(\R \times \tor^3)}
\end{equation}
\end{lemma}


We begin by establishing bilinear Strichartz estimates for frequency-localized functions on $\tor^3,$ then derive bilinear estimates in the $X^{s,b}$-spaces. We follow arguments similar to \cite{BGTzv} with some improvements due to the Strichartz estimates stated in Theorem \ref{KVStrich}. 
 
\begin{prop}[Bilinear Strichartz Estimates]\label{BSE}
Suppose $u_1$ and $u_2 \in L^2(\tor^3),$ have spectra in $[-N_1, N_1]^3$ and  $[-N_2, N_2]^3,$ respectively. Then,
\begin{equation*}
\|e^{it\D}u_1e^{it\D}u_2\|_{L^2_t L^2_x([0,1]\times \tor^3)} \lesssim \min(N_1, N_2)^{\frac{1}{2}}\|u_1\|_{L^2(\tor^3)}\|u_2\|_{L^2(\tor^3)}
\end{equation*}
\end{prop}
\begin{proof} As the time domain $t \in [0,1]$ and and spatial domain $x \in \tor^3$ are fixed, we suppress the domain of the $L^p_t$ and $L^p_x$ norms throughout the proof.
By symmetry, suppose  $N_1 \leq N_2.$ Decompose $\R^3$ into a disjoint collection of cubes $\{\cC_j\},$ each of side length $N_1,$ and observe that $u_1(P_{\cC_j}u_2)$ has spectrum localized in a fixed dilate of $\cC_j.$ Thus we may use almost orthogonality to conclude
\begin{equation*}
\|e^{it\D}u_1e^{it\D}u_2\|_{L^2_t L^2_x} \leq \left(\sum_j \|e^{it\D}u_1e^{it\D}(P_{\cC_j}u_2)\|^2_{L^2_t L^2_x}\right)^{\frac{1}{2}}.
\end{equation*}
By H\"older's inequality, the right hand side is bounded above by
\begin{equation*}
 \|e^{it\D}u_1\|_{L^4_t L^4_x}\left(\sum_j \|e^{it\D}(P_{\cC_j}u_2)\|^2_{L^4_t L^4_x}\right)^{\frac{1}{2}}.
\end{equation*} 
Applying Strichartz estimates to the above upper bound, we conclude
\begin{align*}
\|e^{it\D}u_1e^{it\D}u_2\|_{L^2_t L^2_x}& \lesssim N_1^{\frac{1}{4}}\|u_1\|_{L^2_x}\left(\sum_j N_1^{\frac{1}{2}}\|P_{\cC_j}u_2\|^2_{L^2_x}\right)^{\frac{1}{2}}\\
&\lesssim  N_1^{\frac{1}{2}}\|u_1\|_{L^2_x}\|u_2\|_{L^2_x}
\end{align*}
\end{proof}

The next proposition allows us to move between the previous bilinear Strichartz estimates and bilinear estimates in Bourgain $X^{s,b}$ spaces. The result is contained in \cite{BGTzv}, but the proof is included here for completeness.
\begin{prop}
The following two statements are equivalent:
\begin{enumerate}
\item For $u_1$ and $u_2 \in L^2(\tor^3),$ with spectra in $[-N_1, N_1]^3$ and  $[-N_2, N_2]^3$ respectively,
\begin{equation*}
\|e^{it\D}u_1e^{it\D}u_2\|_{L^2_t L^2_x([0,1]\times \tor^3)} \lesssim \min(N_1, N_2)^s\|u_1\|_{L^2_x(\tor^3)}\|u_2\|_{L^2_x(\tor^3)}
\end{equation*}
\item For any $b > \frac{1}{2}$ and any $v_1, v_2 \in X^{0,b}(\R \times \tor^3)$ with spectra in $[-N_1, N_1]^3$ and $[-N_2, N_2]^3$ respectively,
\begin{equation*}
\|v_1v_2\|_{L^2_t L^2_x(\R \times \tor^3)} \lesssim \min(N_1, N_2)^s\|v_1\|_{X^{0,b}(\R \times \tor^3)}\|v_2\|_{X^{0,b}(\R \times \tor^3)}
\end{equation*}
\end{enumerate}
\end{prop}

\begin{proof}
We show statement (1) implies statement (2) under the assumption that both $v_1$ and $v_2$ are supported on the time interval $(0,1).$ The general case easily follows using a partition of unity argument. By symmetry, suppose $N_1 \leq N_2.$ For $k \in \{1,2\},$ define $V_k:= e^{-it\D}v_k,$ so that we may write
\begin{equation*}
v_k = e^{it\D}V_k.
\end{equation*}
 Use $\Ft$ to denote the Fourier transform in the time variable, and observe
\begin{equation*}
(v_1v_2)(t) = (2\pi)^{-2}\int^{\infty}_{-\infty}\int^{\infty}_{-\infty}e^{it\t + \s}e^{it\D}\Ft V_1(\t)e^{it\D}\Ft V_2(\s)\,d\t d\s.
\end{equation*}
Let us simplify the notation, and write $L^2_tL^2_x:= L^2_t L^2_x([0,1]\times \tor^3).$ By statement (1) of the proposition, we have the estimate
\begin{align}
\|v_1v_2\|_{L^2_t L^2_x} &\leq (2\pi)^{-2}\int^{\infty}_{-\infty}\int^{\infty}_{-\infty}\|e^{it\D}\Ft V_1(\t)e^{it\D}\Ft V_2(\s)\|_{L^2_t L^2_x}\,d\t d\s\notag\\
&\lesssim N_1^s \int^{\infty}_{-\infty}\int^{\infty}_{-\infty} \|\Ft V_1(\t)\|_{L^2_x(\tor^3)}\|\Ft V_2(\s)\|_{L^2_x(\tor^3)}\,d\t d\s .\label{VL2}
\end{align}
Motivated by the observation that for $b>\frac{1}{2},$ $\la \t \ra^{-b} \in L^2_\t(\R),$ we use H\"older's inequality and proceed as follows
\begin{align}
\int^{\infty}_{-\infty}\|\Ft V_1(\t)\|_{L^2_x(\tor^3)}\,d\t &\leq C_b \left(\int^{\infty}_{-\infty}\la \t \ra^{2b} \|\Ft V_1(\t)\|^2_{L^2_x(\tor^3)}\,d\t\right)^{\frac{1}{2}}\notag \\
&= C_b\| V_1(t)\|_{H^b_{t}L^2_x(\R \times \tor^3)}\notag \\
&= C_b \|v_1\|_{X^{0,b}(\R \times \tor^3)}.\label{vXb}
\end{align}
Together, \eqref{VL2} and \eqref{vXb} imply statement (2) when $v_1(t), v_2(t)$ are supported on the time interval $(0,1).$ The general case follows from a standard partition of unity argument.

To see the reverse implication, suppose $u_k \in L^2(\tor^3)$ has spectral support $[-N_k,N_k]^3$ for $k = 1,2,$ and define $U_k(t) := e^{it\D}u_k.$ Let $\psi(t) \in C_0^\infty(\R)$ be supported in the interval $(0,1),$ so that $v_k(t) := \psi(t)U_k(t) \in X^{0,b}(\R).$ The equivalences 
\begin{equation*}
\|v_1v_2\|_{L^2_t L^2_x(\R \times \tor^3)} = \|U_1 U_2 \|_{L^2_t L^2_x((0,1) \times \tor^3)} = \| e^{it\D}u_1e^{it\D}u_2\|_{L^2_t L^2_x((0,1) \times \tor^3)}
\end{equation*}
and 
\begin{equation*}
\|v_k\|_{X^{0,b}(\R \times \tor^3)} = \|e^{-it\D}\psi(t)e^{it\D}u_k\|_{H^b_tL^2_x(\R \times \tor^3)} = C_{\psi}\|u_k\|_{L^2_x(\tor^3)}
\end{equation*}
are all that is needed to see that statement (2) implies statement (1).
\end{proof}


In the next proposition, we establish a range of bilinear estimates using the Bourgain spaces.
\begin{prop}\label{InterpProp}
For any $s>\frac{1}{2}$, 
there is some $\frac{1}{4} < b' < \frac{1}{2}$ such that for any $v_1, v_2 \in X^{0,b}(\R \times \tor^3),$ with spectral support on $[-N_1, N_1]^3$ and $[-N_2, N_2]^3,$ respectively, the following estimate holds:
\begin{equation*}
\|v_1v_2\|_{L^2_t L^2_x(\R \times \tor^3)} \lesssim \min (N_1,N_2)^s\|v_1\|_{X^{0,b'}}\|v_1\|_{X^{0,b'}}
\end{equation*}
\end{prop}
\begin{proof}
Let $v_1$ and $v_2$ have the required spectral support, and suppose $N_1 \leq N_2.$
From the previous lemma and the bilinear Strichartz estimate, for any $\e _0 > 0,$ if $v_1, v_2 \in X^{0,\frac{1}{2} + \e_0},$ then
\begin{equation}\label{interp1}
\|v_1v_2\|_{L^2_t L^2_x(\R \times \tor^3)} \lesssim \min (N_1,N_2)^{\frac{1}{2}}\|v_1\|_{X^{0,\frac{1}{2} + \e_0}}\|v_2\|_{X^{0,\frac{1}{2} + \e_0}}.
\end{equation}
Using H\"older's inequality, Bernstein's inequality, and the inclusion $X^{0,\frac{1}{4}} \hookrightarrow L^4_t L^2_x,$ we derive a second estimate as follows:
\begin{align}
\|v_1v_2\|_{L^2_t L^2_x(\R \times \tor^3)}  &\leq \|v_1 \|_{L^4_t L^{\infty}_x (\R \times \tor^3)} \|v_2 \|_{L^4_t L^2_x (\R \times \tor^3)} \notag \\
&\lesssim N_1^{s
} \|v_1 \|_{L^4_t L^2_x (\R \times \tor^3)} \|v_2 \|_{L^4_t L^2_x (\R \times \tor^3)} \notag \\
& \lesssim N_1^{s
} \|v_1 \|_{X^{0,\frac{1}{4}}}\|v_2 \|_{X^{0,\frac{1}{4}}}, \label{interp2}
\end{align}
for any $s\geq\frac32$.

Interpolating the bounds \eqref{interp1} and \eqref{interp2} gives the desired result.
\end{proof}


We may now prove Lemma \ref{ThreeFuncs}, our key multlinear estimate, using a duality argument combined with a frequency decomposition of the $\uj$ functions.
\begin{proof}[Proof of Lemma \ref{ThreeFuncs}]
Let $(b,b')$ satisfy the hypotheses, with values to be determined later. By duality, we  prove the equivalent estimate: \\*
for any $\uo \in X^{-s,b'}(\R \times \tor^3),$
\begin{equation}
\left\arrowvert\int_{\R}\int_{\tor^3} \ov{\uo} \uf \ov{\us} \ut \,dx dt\right\arrowvert \leq C \|\uo\|_{X^{-s,b'}}\prod_{j=1}^{3}\|u^{(j)}\|_{X^{s,b}(\R \times \tor^3)}.
\end{equation}
By density, we may assume $\uj \in C^{\infty}_0(\R \times \tor^3)$ for $j = 0, 1, 2, 3$ and we will decompose each of these functions into dyadic cubes in Fourier space. 

To this end, we adopt the notation $N_j$ to mean the family of dyadic numbers $\{2^{n_j}\}_{n_j \in \N},$ and the summation $\sum_{N_j} f(N_j)$ indicates to sum over all possible values of $N_j.$ Summing over the collection $\sN$ of all such dyadic decompositions, 
\[\sN = \big\{(N_0, N_1, N_2, N_3)\big\arrowvert N_j \in 2^\N \text{ for  } j = 0, 1, 2, 3\big\},\]
we observe
\begin{equation}
\left\arrowvert\int_{\R}\int_{\tor^3}\ov{\uo} \uf \ov{\us} \ut \,dx dt\right\arrowvert \leq \sum_{\sN}\left\arrowvert\int_{\R}\int_{\tor^3}\ov{\uno} \unf \ov{\uns} \unt \,dx dt\right\arrowvert.
\end{equation}
The integral on the right-hand side is zero unless the two highest frequencies are comparable. Using symmetry, we reduce the sum to two cases.

\emph{Case 1: } Define $\sN_1:=\left\{ N_0 \sim N_1 \geq N_2 \geq N_3\right\} \cap \sN,$ and suppose $s'$ satisfies $\frac{1}{2} < s' < s.$  We use H\"older's inequality, Proposition \ref{InterpProp}, and Bernstein's inequality to show that for some $\frac{1}{4} < b' < \frac{1}{2},$ 
\begin{align}
\sum_{\sN_1} \Bigg\arrowvert\int_{\R}\int_{\tor^3}\ov{\uno} \unf & \ov{\uns}\unt \,dx dt\Bigg\arrowvert \notag \\ 
&\leq \sum_{\sN_1} \| \uno \uns \|_{L^2_t L^2_x}\| \unf \unt \|_{L^2_t L^2_x}\notag \\
&\lesssim \sum_{\sN_1} N_2^{s'}N_3^{s'}\prod_{j=0}^3\|u^{(j)}_{N_j}\|_{X^{0,b'}}\notag\\
&\lesssim \sum_{\sN_1}\frac{N_0^s}{N_1^s}N_2^{s'-s}N_3^{s'-s}\|\uno\|_{X^{-s,b'}}\prod_{j=1}^3\|u^{(j)}_{N_j}\|_{X^{s,b'}}\label{Case1sum}
\end{align}
Noting that $s' - s \leq 0,$ and summing over $N_3 \leq N_2$ using Cauchy-Schwarz, we  bound the expression \eqref{Case1sum} above by
\begin{equation}\label{Case1sum2}
C \|\us\|_{X^{s,b'}} \|\ut\|_{X^{s,b'}} \sum_{N_0 \sim N_1}\|\uno\|_{X^{-s,b'}}\|\unf\|_{X^{s,b'}}.
\end{equation}
We use Cauchy-Schwarz again to sum on $N_0 \sim N_1$ in \eqref{Case1sum2}, concluding
\begin{align}
\sum_{\sN_1} \Bigg\arrowvert\int_{\R}\int_{\tor^3}\ov{\uno} \unf & \ov{\uns}\unt \,dx dt\Bigg\arrowvert \notag \\ 
&\lesssim \|\uo\|_{X^{-s,b'}}\prod_{j=1}^{3}\|u^{(j)}\|_{X^{s,b'}}\label{CSSum1b}.
\end{align}


\emph{Case 2: } Define $\sN_2:=\left\{N_0 \leq N_1 \sim N_2 \geq N_3\right\}\cap \sN$ As in the previous case, for $s'$ satisfying $\frac{1}{2} < s' < s,$  Proposition \ref{InterpProp} guarantees the existence of $b'$ with $\frac{1}{4} < b' < \frac{1}{2}$ such that 
\begin{align}
\sum_{\sN_2} \Bigg\arrowvert\int_{\R}\int_{\tor^3}\ov{\uno} \unf & \ov{\uns}\unt \,dx dt\Bigg\arrowvert \notag \\ 
&\leq \sum_{\sN_2}\| \uno \unf \|_{L^2_t L^2_x}\| \uns \unt \|_{L^2_t L^2_x}\notag \\
&\lesssim \sum_{\sN_2} N_0^{s'}N_3^{s'}\prod_{j=0}^3\|u^{(j)}_{N_j}\|_{X^{0,b'}}\notag \\
&\lesssim \sum_{\sN_2}\frac{N_0^{s' + s}}{N_1^s N_2^s}N_3^{s'-s}\|\uno\|_{X^{-s,b'}}\prod_{j=1}^3\|u^{(j)}_{N_j}\|_{X^{s,b'}},\label{Case2sum1}
\end{align}
where we have used H\"older's inequality for the first line, and Bernstein's inequality for the last. We find upper bounds for last expression above by first summing on $N_0$ and $N_3$, then on $N_1 \sim N_2,$ using Cauchy-Schwarz each time:
\begin{align}
 C \|\uo\|_{X^{-s,b'}}& \|\ut\|_{X^{s,b'}}\sum_{N_1 \sim N_2} \frac{N_1^{s'}}{N_2^s}\|\unf\|_{X^{s,b'}}\|\uns\|_{X^{s,b'}}\notag \\
&\lesssim \|\uo\|_{X^{-s,b'}} \|\ut\|_{X^{s,b'}} \sum_{N_1 \sim N_2} \|\unf\|_{X^{s,b'}}\|\uns\|_{X^{s,b'}}\notag \\
&\lesssim \|\uo\|_{X^{-s,b'}}\prod_{j=1}^{3}\|u^{(j)}\|_{X^{s,b'}}\label{CSSum2b}
\end{align}
Together, \eqref{CSSum1b} and \eqref{CSSum2b} conclude the proof of our lemma.
\end{proof}


\section{Well-posedness of the Cubic NLSS}
\label{sec-5}

We now use a contraction argument on the Duhamel formula for the cubic NLSS to show local well-posedness for initial data in $H^s_{\lu}(\tor^3)$ for $s>\frac{1}{2}$. 
In the defocusing cubic NLSS with initial data in $H^1_{\lu}(\tor^3),$ the local well-posedness combines with the conservation laws to extend the solution for all time. 
\begin{proof}[Proof of Theorem \ref{CubMainThm}]
We begin with the Duhamel formula for the $j$-th equation of the cubic NLSS
\begin{equation*}
\Phi_j(\uu)(t) = e^{it\D}u_{0,j} - i\s \int_{0}^{t} e^{i\D(t - t')}\r u_j(t') dt'
\end{equation*}
where $u_{0,j} = u_j(t=0, x),$ $\uu = \{u_j\}_{j=1}^\infty$ and $\r = \sum_{k \in \N} \l_k|u_k|^2.$ Define the map $\Phi(\uu) := \{ \Phi_j(\uu) \}_{j=1}^\infty.$
Fix $s$ so that $s>\frac{1}{2}$. 
and let $b' = b'(s)$ be the value guaranteed by Proposition \ref{InterpProp}. Choose $b = b(s) > \frac{1}{2}$ so that $b' + b < 1.$
 Suppose $\|\uu_0\|_{H^s_{\lu}(\mathbb{T}^3)} \leq \eta$ for some some $\eta$ to be chosen later. We will show that $\Phi$ is a contraction on the ball
\begin{equation*}
B:=\Big\{\uu \in X^{s,b}_{T,\lu} \cap C_t H^s_{\lu}([0,T]\times \mathbb{T}^3) \, \Big| \, \|\uu\|_{X^{s,b}_{T,\lu}}\leq 2\eta \Big\}
\end{equation*}
for  some $T \leq 1.$


By Proposition \ref{DuhamBds} and Lemma \ref{ThreeFuncs}, we have 
\begin{align*}
\|\Phi_j(\uu)\|_{X^{s,b}_T} &\leq \|e^{it\D} u_{0,j}\|_{X^{s,b}_T} + \left\Vert \int_0^t e^{i(t - t')\D}\r u_j(t')\,dt'\right\Vert_{X^{s,b}_T}\\
&\leq  \|u_{0,j}\|_{H^s(\tor^3)} + CT^{1-b-b'}\sum_k \l_k \||u_k|^2u_j\|_{X^{-s,b'}_T}\\
&\leq \|u_{0,j}\|_{H^s(\tor^3)} + CT^{1-b-b'}\sum_k \l_k \|u_k\|^2_{X^{s,b}_T}\|u_j\|_{X^{s,b}_T} \\
&\leq \|u_{0,j}\|_{H^s(\tor^3)} + CT^{1-b-b'}\|\uu\|^2_{X^{s,b}_{T,\lu}}\|u_j\|_{X^{s,b}_T}
\end{align*}
From the last inequality above, we square both sides, multiply by $\l_j,$ sum on $j$, then take the square root to find
\begin{equation*}
\|\Phi(\uu)\|_{X^{s,b}_{T,\lu}}  \leq \sqrt{2}\|\uu_0\|_{H^s_{\lu}(\tor^3)}+ CT^{1-b-b'}\|\uu\|^3_{X^{s,b}_{T,\lu}}.
\end{equation*}
For $\uu \in B$
\begin{equation*}
\|\Phi(\uu)\|_{X^{s,b}_{T,\lu}}  \leq \sqrt{2}\eta+ CT^{1-b-b'}\eta^3
\end{equation*}
and the right side is bounded above by $2\eta$ for $T$ small enough depending on $C$ and $\eta.$ 


For the contraction argument on $B$, we first observe that
\begin{align}
\sum_k \l_k\big\|\,|u_k|^2 u_j & - |v_k|^2v_j\big\|_{X^{-s,b'}_T}\notag \\
&\leq \sum_k \l_k \|u_k\|^2_{X^{s,b}_T}\|u_j - v_j \|_{X^{s,b}_T}\notag\\
&\qquad\qquad+\sum_k \l_k  \|u_k - v_k\|_{X^{s,b}_T}\Big(\|u_k\|_{X^{s,b}_T} + \|v_k\|_{X^{s,b}_T}\Big)\|v_j\|_{X^{s,b}_T}\notag\\
&\leq \|\uu\|^2_{X^{s,b}_{T,\lu}}\|u_j - v_j \|_{X^{s,b}_T}\notag\\
&\qquad\qquad +\| \uu- \vu\|_{X^{s,b}_{T,\lu}}\Big(\| \uu \|_{X^{s,b}_{T,\lu}}  + \| \vu \|_{X^{s,b}_{T,\lu}}  \Big)\|v_j\|_{X^{s,b}_T}\label{CubicContract}
\end{align}
where we have used Cauchy-Schwarz twice. We combine the above argument with Ginibre's estimate, then square, multiply by $\l_j$, sum on $j,$ then take the square root to find
\begin{equation*}
\|\Phi(\uu)- \Phi(\vu)\|_{X^{s,b}_{T,\lu}} \lesssim T^{1-b-b'} \| \uu- \vu\|_{X^{s,b}_{T,\lu}}\left(\| \uu \|_{X^{s,b}_{T,\lu}}  + \| \vu \|_{X^{s,b}_{T,\lu}}  \right)^2.
\end{equation*}
For small enough $T$ depending on $s,$ $\|\uu_0\|_{H^s_{\lu}(\tor^3)},$ and the implicit constant, $\Phi$ is a contraction on $B$ in the $X^{s,b}_{\lu}$ norm, and we obtain a unique solution to the Cauchy problem on $[0,T).$ Continuous dependence on initial data is obtained using a similar argument, and we conclude the cubic NLSS is locally well-posed in $H^s_{\lu}(\tor^3)$ for $s>\frac{1}{2}$. 


Now consider the defocusing cubic NLSS with initial data $\uu_0 \in H^1_{\lu}(\tor^3).$ Recall, the conserved energy is 
\[E_{\lu} = \frac{1}{2}\sum_k \l_k \int |\n u_k|^2\,dx + \frac{1}{4}\| \r \|^2_{L^2_x}.\] 
By H\"older's inequality and Sobolev embedding, we have
 \[ \|\r \|_{L^2(\tor^3)}  \leq \sum_k \l_k \|u_k \|_{L^4(\tor^3)}^2 \lesssim \sum_k \l_k \|u_k \|_{L^6(\tor^3)}^2 \lesssim \|\uu\|_{H^1_{\lu}(\tor^3)}^2,\]
so that 
\begin{align*}
\|\uu(t)\|^2_{H^1_{\lu}(\tor^3)}&\leq M(\uu(t)) + 2E(\uu(t))=M(\uu_0) + 2E(\uu_0)\\
&= \| \uu(0)\|^2_{H^1_{\lu}(\tor^3)} + \frac{1}{2} \|\r(0)\|_{L^2_x(\tor^3)}^2 \\
& \leq \|\uu(0)\|^2_{H^1_{\lu}(\tor^3)} + C\|\uu_0\|_{H^1_{\lu}(\tor^3)}^4, 
\end{align*}
 Thus for some $T' < T,$ depending on the constant in the above upper bound, we may repeat the local well-posedness argument on intervals of length $T'$ indefinitely.
\end{proof}

\section{Well-Posedness of the Quintic NLSS on $\tor^3$}
\index{Well-Posedness of the Quintic NLSS on $\tor^3$%
@\emph{Well-Posedness of the Quintic NLSS on $\tor^3$}}%
\label{sec-6}


The quintic nonlinear Schr\"odinger system is given by
\begin{equation}\label{QNLSS}
\begin{cases}
& i\pd_t u_j =-\D u_j+ \s\r^2 u_j, \qquad j \in \mathbb{N}\\
&u_j(0,x) = u_{j,0}(x), \qquad x \in \tor^3,
\end{cases}
\end{equation}
where $\s \in \{-1,1\},$ $\lu \in \ell^1_+$, and $\r(t,x) = \sum \l_j |u_j(t,x)|^2.$ The system has the conserved quantities of mass and energy, given by
\begin{align}\label{QuinCons}
M_{\lu}(\uu)&:=\sum_k \l_k \|u_k\|_{L^2(\tor^3)}^2 = \|\uu\|_{L^2_{\lu}}^2 \\
E_{\lu}(\uu) &:=\frac{1}{2}\sum_k \l_k \|\n u_k\|_{L^2(\tor^3)}^2 +\s \frac{1}{6}\int_{\tor^3} \r^3\,dx.
\end{align}

In this chapter, we prove the following theorem:
\begin{theorem}[Local and global well-posedness of the quintic NLSS]\label{QuinMainThm}
Let  $\lu \in \ell^1_+,$ and suppose $\uu_0 \in H_{\lu}^1(\tor^3).$ There exists a time $T$ depending on $\uu_0$ such that the system \eqref{QNLSS} is locally well-posed for $t \in [0,T)$. Moreover, there exists $\eta>0$ such that if $\| \uu_0\|_{H_{\lu}^1(\mathbb{T}^d)} \leq \eta,$ then the solution is global in time.
\end{theorem}
As in the case of the quintic NLS equation, the time of existence depends on the function itself, and global well-posedness holds for initial data with sufficiently small $H^1_{\lu}$ norm. 

We will use some of the same tools as were used in the cubic case, namely establishing multilinear estimates using frequency decompositions and the Strichartz estimate \ref{Strich1}. However, following \cite{kv} and \cite{HTaTzv}, we will use the function spaces $X^s$ and $Y^s$ in our analysis, similar to the $X^{s,b}$ spaces, as they are well-suited the study of the energy-critical system.


\section{Relevant Function Spaces and Their Properties}

The definitions of the  $X^s$ and $Y^s$ spaces are based on underlying $U^p$ and $V^p$ spaces. We present an overview of this construction, and state some of the properties of these function spaces that we require for our analysis. For a thorough treatment of these spaces, we refer the interested reader to \cite{HTaTzv}.

We construct the $X^s$ and $Y^s$ on finite time intervals, and as in the previous chapter, our norms will be restriction norms on the given time interval. Let $H$ be a separable Hilbert space over $\mathbb{C}$, and $[0,T]$  a finite time interval. Let $\mathscr{T}$ be the set of partitions of the interval $[0,T],$ that is, $\displaystyle{\{t_j\} _{j=0}^M} \in \mathscr{T}$ whenever $0 = t_0 < t_1 < ...< t_M \leq T$ for some finite $M.$ For functions $u: [0,T) \rightarrow H,$ we define $u(T):=0$ at the endpoint of the interval. 


\begin{defin}
 A $U^p$-atom, $1 \leq p < \infty$ is a function $a:[0,T) \rightarrow H$ of the form
\begin{equation*}
a = \sum_{j=1}^M \chi_{[t_{j-1},t_j)}\phi_{j-1}
\end{equation*}
where $M < \infty,$ $\displaystyle{\{t_j\}_{j=0}^M} \in \mathscr{T}$ and the sequence $\{\phi_j\} \subset H$ satisfies $\sum_{j=0}^M \| \phi_j\|_H^p =1.$ Define $U^p([0,T);H)$ to be the space of all functions that may be represented in the form
\begin{equation*}
u = \sum_{k = 1}^\infty \mu_k a_k
\end{equation*}
where $\{\mu_k\} \in \ell^1(\mathbb{C})$ and  $\{a_k\}$ are $U^p$-atoms. $U^p$ is a Banach space with the norm
\begin{equation*}
\|u\|_{U^p([0,T);H)} :=\inf \Big\{\sum_{k=1}^{\infty}|\mu_k| \big\arrowvert u = \sum_{k= 1}^\infty \mu_k a_k \text{  with } \{\mu_k\} \in \ell^1(\mathbb{C}) \text{ and } U^p\text{-atoms } a_k \Big\}
\end{equation*}
\end{defin}
\begin{defin}
The space $V^p([0,T);H),$ $1 \leq p < \infty,$ is the space of all functions $v:[0,T) \rightarrow H$ such that
\begin{equation*}
\| v\|_{V^p([0,T);H)} := \sup_{ \{t_k\} \in \mathscr{T}}\Big(\sum_{k=1}^M \| v(t_k) - v(t_{k-1})\|_H^p \Big)^{1/p} < \infty .
\end{equation*} 
Define $V_{rc}^p$ to be the closed subspace of $V^p$ consisting of right-continuous functions $v(t)$ such that $v(0) = 0.$  $V^p_{rc}$ is a Banach space under the above norm.
\end{defin}


\begin{defin}
For $s \in \R,$ we define the spaces $X^s([0,T))$ and $Y^s([0,T))$ as the spaces of all functions $u: [0,T) \rightarrow H^s(\tor^d)$ such that  for every $\xi \in \mathbb{Z}^d,$  $e^{it|\xi|^2}\widehat{u(t)}(\xi)$ is in $U^2([0,T);\C 
	)$ and $V^2_{rc} ([0,T);\C
	),$  respectively, with finite norms
\eqnn
\|u\|_{X^s([0,T))}:=&\left( \displaystyle{\sum_{\xi \in \mathbb{Z}^d}}\la \xi \ra^{2s}\|\widehat{e^{-it\D}u(t)}(\xi) \|^2_{U^2}\right)^{1/2}\\
\|u\|_{Y^s([0,T))}:=&\left(\displaystyle{ \sum_{\xi \in \mathbb{Z}^d}}\la \xi \ra^{2s}\|\widehat{e^{-it\D}u(t)}(\xi) \|^2_{V^2}\right)^{1/2}.
\eeqnn
\end{defin}


\begin{rem}\label{useful}
 We record the following properties of $X^s$ and $Y^s$:
\begin{enumerate}[label = (\arabic*), labelindent=\parindent]
\item We have the continuous embeddings $X^s \hookrightarrow Y^s$ and $X^s \hookrightarrow C_tH^s_x$
\item  The $X^s$ and $Y^s$ spaces scale like $L^{\infty}_tH^s_x$ and have the same Fourier-based properties, including Bernstein inequalities and square summability.
\item Proposition 2.11 in \cite{HTaTzv} gives
\begin{equation*}
\left\Vert \int_0^t e^{i(t-t')\D}F(t')\,dt'\right\Vert_{X^s([0,T))} \lesssim \| F \|_{ L_t^1H_x^s([0,T) \times (\tor^3))}
\end{equation*}
\item For $p>\frac{10}{3}$ the Strichartz estimate \ref{Strich1} gives 
\begin{align}
\| P_{\leq N} u\|_{L^p_{t,x}([0,T) \times \tor^3)}& \lesssim N^{\frac{3}{2} - \frac{5}{p}} \|e^{-it\D} P_{\leq N} u \|_{U^p\left([0,T); L^2(\tor^3)\right)}\notag\\
& \lesssim  N^{\frac{3}{2} - \frac{5}{p}} \|P_{\leq N} u \|_{Y^0([0,T))}.\label{Strich2}
\end{align}
\end{enumerate}
\end{rem}


\section{Nonlinear Estimates}
We begin this section by stating the following proposition from \cite{HTaTzv}, which allows us to estimate the nonlinear term of the Duhamel formula using a dual formulation.
\begin{prop}\label{Dual}
Let $s \geq 0$ and $T > 0.$ If $F(t,x) \in L^1_tH^s_x([0,T) \times \tor^3),$ then $ \int_0^t e^{i(t-t')\D}F(t')\,dt' \in X^s([0,T)),$ and 
\begin{equation*}
\left\Vert \int_0^t e^{i(t-t')\D}F(t')\,dt'\right\Vert_{X^s([0,T))} \leq \sup_{v \in Y^{-s}([0,T)), \|v\|_{Y^{-s}}=1}\left| \int_0^T\int_{\tor^3} F(t,x)\ov{v(t,x)}\,dxdt\right|. \label{XsDual}
\end{equation*}
\end{prop}
 We will use this dual formulation, combined with a frequency decomposition argument similar to the argument in Chapter 2 to prove the following lemma:
 
\begin{lemma}\label{NonlinEst}
For $\lu \in \ell^1_+$ and a fixed value of $T$ satisfying $0< T \leq 1,$ there is a constant $C > 0$ (which does not depend on $T$) such that for any quintuple $u^{(j)}\in X^1([0,T))$, $j=1,\dots,5$,
\begin{equation}\label{eq-X1-bd-1}
\|  \int_0^t e^{i(t-s)\D}\Big(\prod_{j=1}^5 u^{(j)}(s)\Big) ds\|_{X^1([0,T])} \leq C \prod_{j=1}^5\|u^{(j)}\|_{X^1([0,T])}  .
\end{equation}
In particular,
\begin{equation}\label{eq-X1-bd-2}
\|  \int_0^t e^{i(t-s)\D}\r^2u_j(s) ds\|_{X^1([0,T])} \leq C \|\uu\|^4_{X^1_{\lu}([0,T])} \| u_j \|_{X^1([0,T])}.
\end{equation}
\end{lemma}
\prf
Fix $N \geq 1,$ and note $P_{\leq N}[ \r^2u_j] \in L^1([0,T];H^1(\mathbb{T}^3)).$ By duality, we have from \cite{HTaTzv}
\begin{multline*}
\left\Vert \int_0^t e^{i(t-s)\D} P_{\leq N}[ \r^2u_j(s)]\,ds\right\Vert_{X^1([0,T])}\\
\leq \sup_{\| \tilde{v} \|_{Y^{-1}([0,T])}=1}\left\vert\int_0^T \int_{\mathbb{T}^3} P_{\leq N}[ \r^2u_j(t,x)]\overline{\tilde{v}(t,x)}dxdt\right\vert.
\end{multline*}
 Letting $v:=\overline{P_{\leq N}\tilde{v}},$  
\begin{equation}
\left\vert\int_0^T \int_{\mathbb{T}^3} \r^2u_j v\,dxdt\right\vert \leq \sum_{k,l} \l_k \l_l \left\vert \int_0^T \int_{\mathbb{T}^3}|u_k|^2|u_l|^2u_j v\,dxdt\right\vert
\end{equation}
Observe that our problem reduces to finding an upper bound for the double integral on the right hand side for fixed $k$ and $l.$ To that end, for $i \in \{1,2,3,4,5\}$, let $u^{(i)}$ be one of the collection $\{u_k, \overline{u_k}, u_l, \overline{u_l}, u_j\}$ so that the list is exhausted as $i$ varies from $1$ to $5.$ We write each factor as a sum of dyadic frequency projections, that is,
\begin{equation*}
\left\vert \int_0^T \int_{\mathbb{T}^3}|u_k|^2|u_l|^2u_j v\,dxdt\right\vert \leq \sum_{\mathcal{N}}\left\vert\int_0^T\int_{\tor^3} \vno \unf\uns\unt\unfo\unj \,dxdt\right\vert
\end{equation*}
Where $\mathcal{N} = \big\{N_{i} \in 2^{\N}, \text{ for } i \in \{0,1,2,3,4,5\}\big\}$

Note that the integral on the right hand side above is only nonzero when the two largest frequencies are comparable. By this fact and symmetry, we may break $\mathcal{N}$ into two cases where, the two largest frequencies are $N_0\sim N_5$ and $N_5 \sim N_1.$  In the analysis of each case, we adopt the abbreviated notations $\|\cdot\|_p:=\| \cdot\|_{L^p_{t,x}([0,T],\tor^3)},$ $\| \cdot \|_{Y^s} :=\| \cdot \|_{Y^s([0,T])},$ and we use a similar abbreviation for the $X^s$ norm.


\emph{Case 1:} $\sN_1 =\{N_0 \sim N_5 \geq N_1\geq N_2 \geq N_3 \geq N_4\}\cap \sN.$ Subdivide $\Z^3$ into cubes $\cC_m$ of size $N_1,$ and write $\cC_m \sim \cC_n$ if the set $\cC_m +\cC_n$ overlaps the Fourier support of $P_{\leq 2 N_1}.$ Note that here are a bounded number of $\cC_m \sim\cC_n$ for a given $\cC_n.$ Using H\"older's inequality, Strichartz estimates, and Bernstein's inequalities, we have
\begin{align}
\sum_{\sN_1}&\int_0^T\int_{\tor^3}\left\vert \vno \unj \unf \uns\unt\unfo \right\vert\,dxdt \label{CaseQN1}\\
&\lesssim\sum_{\sN_1}\, \sum_{\cC_m\sim \cC_n} \| P_{\cC_m}\vno \|_{4} \|P_{\cC_n}\unj\|_{4} \| \unf \|_{4} \| \uns \|_{4} \| \unt \|_{\infty} \| \unfo \|_{\infty}\notag \\
&\lesssim \sum_{\sN_1}\, \sum_{\cC_m\sim \cC_n} N_1^{3/4}N_2^{1/4}N_3^{3/2}N_4^{3/2} \| P_{\cC_m}\vno \|_{Y^0} \|P_{\cC_n}\unj \|_{Y^0}\displaystyle{ \prod_{i=1}^{4}} \|u^{(i)}_{N_i} \|_{Y^0}\notag\\
& \lesssim\sum_{\sN_1}\, \sum_{\cC_m\sim\cC_n} \frac{N_0 N_3^{1/2}N_4^{1/2}}{N_5 N_1^{1/4}N_2^{3/4} } \| P_{\cC_m}\vno \|_{Y^{-1}} \|P_{\cC_n}\unj\|_{Y^1}\displaystyle{ \prod_{i=1}^{4}} \|u^{(i)}_{N_i} \|_{Y^1}\label{CaseQN1sum1}
\end{align}
We apply the Cauchy-Schwarz inequality, then sum on $N_4$ for $N_4 \leq N_3.$ We then repeat this process for $N_3 \leq N_2$ to see that \eqref{CaseQN1sum1} is controlled by
\begin{multline*}
 \|\ut \|_{Y^1} \|\ufo \|_{Y^1}\left( \sum_{N_0 \sim N_5}\, \sum_{\cC_m\sim \cC_n} \| P_{\cC_m}\vno \|_{Y^{-1}}\|P_{\cC_n}\unj\|_{Y^1} \right)\\
\times\left(\sum_{N_1 \geq N_2}\left(\frac{N_2}{N_1} \right)^{1/4} \| \unf \|_{Y^1} \| \uns \|_{Y^1}\right)
\end{multline*}
Using Cauchy-Schwarz to find an upper bound for each of the sums, we first sum on the set $N_2 \leq N_1 \leq N_5,$ then on the set $\cC_m \sim \cC_n,$ and finally on the set $N_5\sim N_0.$ We find that the previous expression is bounded above by  
\begin{equation*}
C\| v \|_{Y^{-1}} \|u_j \|_{Y^1}\displaystyle{\prod_{i=1}^4}\|u^{(i)}\|_{Y^1},
\end{equation*}
for some constant $C > 0.$ The embedding $X^1 \hookrightarrow Y^1$ proves that 
\begin{equation}
\sum_{\sN_1}\int_0^T\int_{\tor^3}\left\vert \vno \unf \uns\unt\unfo\unj  \right\vert\,dxdt \leq 
C\| v \|_{Y^{-1}} \prod_{j=1}^5\|u^{(j)}\|_{X^1}
\label{CaseQn1sum2}
\end{equation}
The implicit constant arises from the use of the Strichartz estimates, Bernstein inequalities, and the embbedding  $X^1 \hookrightarrow Y^1,$ thus is independent of $T.$


\emph{Case 2:} $\sN_2:=\left\{N_0 \leq N_5 \sim N_1 \geq N_2 \geq N_3 \geq N_4\right\}\cap \sN.$ In this case, it is not necessary to subivide into cubes. From H\"older's inequality, Strichartz estimates, and Bernstein's inequality, we find 
\begin{align}
\sum_{\sN_2}&\int_0^T\int_{\tor^3}\left\vert \vno \unj \unf \uns\unt\unfo \right\vert\,dxdt \label{CaseQN2}\\
&\sum_{\sN_2} \| \vno \|_{4} \|\unj\|_{4} \| \unf \|_{4} \| \uns \|_{4} \| \unt \|_{\infty} \| \unfo  \|_{\infty} \notag\\
&\lesssim \sum_{\sN_2}\, \left(N_0 N_5 N_1 N_2\right)^{\frac{1}{4}}\left(N_3 N_4 \right)^{\frac{3}{2}} \|\vno \|_{Y^0} \|\unj\|_{Y^0}\displaystyle{ \prod_{i=1}^{4}} \|u^{(i)}_{N_i} \|_{Y^0}\notag\\
&\lesssim \sum_{\sN_2}\, \frac{N_0^{\frac{5}{4}}N_3^{\frac{1}{2}}N_4^{\frac{1}{2}}}{N_5^{\frac{3}{4}}N_1^{\frac{3}{4}}N_2^{\frac{3}{4}}}\|\vno \|_{Y^{-1}}\|\unj\|_{Y^1} \prod_{i=1}^{4} \|u_{N_i}^{(i)} \|_{Y^1}\label{CaseQN2sum1}
\end{align}
Using Cauchy-Schwarz for each sum, we sum in the order $N_4 \leq N_3,$ $N_3 \leq N_2,$ $N_2 \leq N_1,$ and $N_0 \leq N_5.$ Thus there is a constant $C >0$ such that \eqref{CaseQN2sum1} is bounded above by
\begin{align*}
C \|v \|_{Y^{-1}}&\| \us \|_{Y^1}\| \ut \|_{Y^1} \| \ufo \|_{Y^1}\displaystyle{\sum_ {N_5 \sim N_1}}\frac{N_5^{\frac{1}{2}}}{N_1^{\frac{1}{2}}} \|\unj\|_{Y^1}\| \unf  \|_{Y^1}\\
&\lesssim \|v \|_{Y^{-1}} \|u^{(5)} \|_{Y^1} \displaystyle{\prod_{i=1}^4}\|u^{(i)}\|_{Y^1}.
\end{align*}
Finally, we again use the embedding  $X^s \hookrightarrow Y^s$ to conclude
\begin{eqnarray}\label{CaseQN2sum2} 
	\sum_{\sN_2}\int_0^T\int_{\tor^3}\left\vert \vno  \unf \uns\unt\unfo\unj \right\vert\,dxdt 
	\leq
	C\| v \|_{Y^{-1}} \prod_{j=1}^5\|u^{(j)}\|_{X^1},
\end{eqnarray}
where the implicit constant $C > 0$ arises in the same way as in Case 1.
Together, the bounds \eqref{CaseQn1sum2} and \eqref{CaseQN2sum2} yield
\begin{eqnarray}
	\left\vert \int_0^T \int_{\mathbb{T}^3}v\Big(\prod_{j=1}^5 u^{(j)}(t)\Big)\,dxdt\right\vert 
 	\leq  
	C\| v \|_{Y^{-1}} \prod_{j=1}^5\|u^{(j)}\|_{X^1}  \,.
\end{eqnarray}
Recalling that $v=\overline{P_{\leq N}\tilde{v}}$ where $\|\tilde{v}\|_{Y^{-1}}=1,$ and letting $N \rightarrow \infty$, 
we infer that the asserted bound \eqref{eq-X1-bd-1} holds.

Recalling that $u^{(i)}$, $i=1,\dots,5$, enumerates the collection $\{u_k, \overline{u_k}, u_l, \overline{u_l}, u_j\}$, we have
\begin{eqnarray}
	\left\| \int_0^T \int_{\mathbb{T}^3}|u_k|^2|u_l|^2u_j\,dxdt\right\|_{X^1} 
 	\leq 
	 C \|u_j \|_{X^1} \|u_k\|^2_{X^1} \| u_l\|^2_{X^1} \,.
\end{eqnarray}
Multiplying the above inequality by $\l _k \l _j$ and summing on $k,l$, we obtain
\begin{equation*}
\left\|\int_0^T \int_{\mathbb{T}^3} \r^2u_j \,dxdt\right\|_{X^1}  \leq C  \| \uu\|^4_{X^1_{\lu}} \|u_j \|_{X^1}
\end{equation*}
where $C >0 $ does not depend on time. This proves \eqref{eq-X1-bd-2}.

\endprf
 

\section{Proof of Main Result for the Quintic Case}
\prf[Proof of Theorem \ref{QuinMainThm}]
\label{sec-9}


We first show local well-posedness for small initial data.
 The Duhamel formula for the $j$-th equation in the quintic NLS system is given by
\begin{equation*}
\Phi_j(\uu)(t) = e^{it\D}u_{0_j} - i \s \int_0^t e^{i(t-t')\D}\r^2u_j(t')\, dt'.
\end{equation*}
Define the map $\Phi(\uu) := \{ \Phi_j(\uu) \}_{j=1}^\infty.$ Suppose $\|\uu_0\|_{H^1_{\lu}(\mathbb{T}^3)} \leq \eta$ for some small $\eta$ to be chosen later. We will show that $\Phi$ is a contraction on the ball
\begin{equation*}
B_1:=\Big\{\uu \in X^1_{\lu}([0,1]) \cap C_t H^1_{\lu}([0,1]\times \mathbb{T}^3) \, \Big| \, \|\uu\|_{X^1_{\lu}([0,1])}\leq 2\eta \Big\}
\end{equation*}
under the $X^1_{\lu}([0,1])$ norm. As we proceed, each $X^1$ and $X^1_{\lu}$ norm will be over the interval $[0,1]$ and each $H^1$ and $H^1_{\lu}$ norm will be over $\mathbb{T}^3.$

By Lemma \ref{NonlinEst} we have 
\begin{equation*}
\| \Phi_j(\uu) \|_{X^1} \leq \|u_{0_j}\|_{H^1} + C \|\uu\|^4_{X^1_{\lu}} \| u_j \|_{X^1},
\end{equation*}
square each side, multiply by $\l_j,$ sum on $j,$ and take the square root to find 
\begin{equation}\label{contract}
\|\Phi (\uu) \|_{X^1_{\lu}} \leq \sqrt{2}\|\uu_0\|_{H^1_{\lu}} + C \|\uu\|^5_{X^1_{\lu}} 
\end{equation}
For $\uu \in B_1,$ we have
$\|\Phi (\uu) \|_{X^1_{\lu}} \leq \sqrt{2}\eta + C(2\eta)^5 $.
The right hand side is bounded by $2\eta$ if $\eta$ is sufficiently small, thus $\Phi$ maps the ball $B_1$ to itself.

Next we show that $\Phi$ is a contraction on $B_1.$ Let $\uu, \vu \in B_1$ and consider $\| \Phi(\uu-\vu) \|_{X^1_{\lu}}.$ We use arguments similar to those leading to equation \eqref{CubicContract} in the cubic case, to show
\begin{equation}\label{QuinticContract}
\|\Phi_j(\uu - \vu) \|_{X^1} \lesssim \big( \| \uu \|_{X^1_{\lu}} + \| \vu \|_{X^1_{\lu}} \big)^3\| \uu - \vu \|_{X^1_{\lu}}\| u_j \|_{X^1} +  \| \vu \|^4_{X^1_{\lu}} \|u_j - v_j \|_{X^1}
\end{equation}
We then square the above estimate, multiply by $\l_j$ and take the square root to find
\begin{equation*}
\|\Phi (\uu - \vu) \|_{X^1_{\lu}}  \leq  C\| \uu - \vu \|_{X^1_{\lu}}\left( \| \uu\|_{X^1_{\lu}}+  \| \vu \|_{X^1_{\lu}} \right)^4
\end{equation*}
Thus, for $\uu, \vu \in B_1,$ we have
\begin{align*}
\|\Phi (\uu - \vu) \|_{X^1_{\lu}}&\leq C \| \uu - \vu \|_{X^1_{\lu}}\left(4\eta\right)^4\\
& \leq \frac{1}{2} \| \uu -\vu \|_{X^1_{\lu}}
\end{align*}
for $\eta$ sufficiently small. By the contraction mapping principle, we obtain a solution $\uu$ on the time interval $[0,1].$


The global well-posedness for the case of small initial data is obtained from the conserved mass and energy for the energy-critical NLS system:
\begin{equation*}
\mathcal{M}(\uu) = \| \uu \|^2_{L^2_{\lu}(\tor^3)} \qquad
\mathcal{E}(\uu) = \frac{1}{2} \sum_{j}\big( \l_j \int_{\mathbb{T}^3} \vert  \n u_j \vert^2\,dx \big)  + \s \frac{1}{6}\int_{\mathbb{T}^3}\r^3\,dx 
\end{equation*}
In the defocusing case, $\s = 1,$ we may expand $\r_0$ and apply the Sobolev embedding $H^1 \hookrightarrow L^6,$  to find 
\begin{align*}
\| \uu(t) \|^2_{H^1_{\lu}} &\leq \mathcal{M}(\uu(t)) + 2\mathcal{E}(\uu(t)) =\mathcal{M}(\uu_0) +2\mathcal{E}(\uu_0) \\
&\leq \| \uu_0 \|^2_{H^1_{\lu}} + \frac{1}{3}C\| \uu_0 \|^6_{H^1_{\lu}}
\end{align*}
 For $\| \uu_0 \|_{H^1_{\lu}}$ sufficiently small, we may ensure $\| \uu(t) \|_{H^1_{\lu}} \leq \eta$ throughout its time of existence. We may continue iterating the previous local well-posedness arguments indefinitely to obtain global well-posedness.
 
For global in time solutions in the focusing case, $\s = -1,$ we again use the conservation of mass and energy, combined with a continuity argument as follows. First, we observe
\begin{equation}
\| \uu(t) \|^2_{H^1_{\lu}} = \mathcal{M}(\uu(0)) + 2\mathcal{E}(\uu(0)) + \frac{1}{3}\| \r(t) \|^3_{L^3_{\lu}}
\end{equation}
Expand $\r(t)$ and again use Sobolev embedding to obtain the inequality that we will use for the continuity argument:
\begin{equation}
\| \uu(t) \|^2_{H^1_{\lu}} \leq  \| \uu_0 \|^2_{H^1_{\lu}} + \frac{1}{3}C\| \uu_0 \|^6_{H^1_{\lu}} + \frac{1}{3}C\| \uu(t) \|^6_{H^1_{\lu}}
\end{equation}

Define $f(x) = x - (1/3)Cx^3$ so that we have $f(\|\uu(t)\|_{H^1_{\lu}}^2) 
\leq \| \uu_0 \|^2_{H^1_{\lu}} + (1/3)C\| \uu_0 \|^6_{H^1_{\lu}}$ on the time interval $[0,1].$ On the interval $I:=[0,C^{-1/2}]$ the function $f(x)$ increases from $0$ to a maximum value of $(2/3)C^{-1/2}$ and satisfies $f(x) \geq (2/3)x$ for all $x \in I.$ 

Set $\eta_0^2 = \min \{ (2/3)C^{-1/2}, (2/3)\eta^2 \},$ and consider initial data satisfying 
\[ \| \uu_0 \|^2_{H^1_{\lu}} + \frac{1}{3}C\| \uu_0 \|^6_{H^1_{\lu}} \leq \eta_0^2.\]
 We then have $f(\|\uu(t)\|_{H^1_{\lu}}^2)\leq  (2/3)C^{-1/2}.$ The continuity of $\|\uu(t)\|_{H^1_{\lu}}$ in $t$ implies $\|\uu(t)\|^2_{H^1_{\lu}} \in I$ for $t \in [0,1],$ so that 
 \[ \|\uu(t)\|^2_{H^1_{\lu}}\, \leq \, \frac{3}{2}f(\|\uu(t)\|^2_{H^1_{\lu}})\,\leq \,\frac{3}{2}\eta_0^2\,\leq \,\eta^2\] 
 for all $t \in [0,1].$ Therefore we may iterate the local well-posedness argument to obtain global well-posedness for sufficiently small initial data. 
 

We now turn to the task of showing local well-posedness for large initial data.
Let $\Vert\uu_0 \Vert_{H^1_{\lu}(\tor^3)} \leq A$ for some $0 < A < \infty.$ Let $\delta = \delta(A) >0$ (to be chosen later) and $N = N(\uu_0,\delta) \geq 1$ such that $\| P_{>N}\uu_0\|_{H^1_{\lu}(\tor^3)} \leq \delta.$

For some $T = T(\uu_0),$ the mapping $\Phi(\uu)$ is a contraction on the ball
\begin{equation*}
B_2:=\Big\{\uu \in X^1_{\lu}([0,T))\cap C^1_t H^1_{\lu}\left([0,T) \times \tor^3\right) \Big| \|\uu\|_{X^1_{\lu}[0,T)}\leq 2A,\, \|P_{>N}\uu\|_{X^1_{\lu}[0,T)}\leq 2\delta \Big\}
\end{equation*}
under the $X^1_{\lu}$-norm. In what follows, norms in time will be taken over the interval $[0,T)$ and norms in space are on the domain $\tor^3.$ We use $C$ to denote any positive constant which does not depend on $T.$


To prove that $\Phi$ maps $B_2$ to itself, we write
\begin{equation}
	\Phi_{NL}(\uu) := 	-\sigma\int_0^t e^{i(t-t')\D}\r^2\uu(t')
\end{equation}
for its nonlinear part, and observe that
\begin{eqnarray}
\|P_{>N}\Phi(\uu)\|_{X^1_{\lu}} &\leq& \sqrt2\|P_{>N}\uu_{0}\|_{X^1_{\lu}} + \sqrt2\|P_{>N}\Phi_{NL}(\uu)\|_{X^1_{\lu}}
	\nonumber\\
	&\leq&\sqrt2\eta + \sqrt2 \|\Phi_{NL}(\uu)\|_{X^1_{\lu}}.
\end{eqnarray}
Clearly, $\Phi_{NL}(\uu)$ is quintic in $\uu=P_{\leq N}\uu + P_{>N}\uu$, and we decompose it into
\begin{eqnarray}
	\Phi_{NL}(\uu) = \Phi_{NL}^{(1)}(P_{\leq N}\uu,P_{>N}\uu)+\Phi_{NL}^{(2)}(P_{\leq N}\uu, P_{>N}\uu)
\end{eqnarray}
where $\Phi_{NL}^{(1)}$ is at least quadratic in $P_{>N}\uu$, and $\Phi_{NL}^{(2)}$ is at least quartic in $P_{\leq N}\uu$. Then, \eqref{eq-X1-bd-1} and the argument used to obtain \eqref{eq-X1-bd-2} imply that
\begin{eqnarray}
	\|\Phi_{NL}^{(1)}(P_{\leq N}\uu,P_{>N}\uu)\|_{X^1_{\lu}} 
	&\leq&
	C \|\uu\|_{X^1_{\lu}}^3\|P_{>N}\uu\|_{X^1_{\lu}}^2
	\nonumber\\
	&\leq&
	C_1 A^3 \eta^2 \,.
\end{eqnarray}
To bound $\Phi_{NL}^{(2)}$, we use the notation
\begin{equation}
		\|\uu\|_{L^q_t L^p_{\lu}} := \Big(\int_0^t ds \Big( \sum_j \lambda_j \|u_j(s)\|_{L^p_x}^2\Big)^{\frac q2}\Big)^{\frac1q} \,.
\end{equation}
Then, applying H\"older, we get
\begin{eqnarray}
	\lefteqn{
	\|\Phi_{NL}^{(2)}(P_{\leq N}\uu,P_{>N}\uu)\|_{X^1_{\lu}} 
	}
	\nonumber\\
	&\leq&
	C_1 \|\uu\|_{L^\infty_t H^1_{\lu}}\|P_{\leq N }\uu\|_{L^4_t L^\infty_{\lu}}^4 
	+C_1 N \|\uu\|_{L^\infty_t L^6_{\lu}} \|P_{\leq N }u\|_{L^4_t L^{12}_{\lu}}^4
\end{eqnarray}
where the first term on the r.h.s. bounds the expression obtained from the derivative in the definition of $X^1_{\lu}$ acting on $P_{>N}\uu$, and the second term from it acting on $P_{\leq N}\uu$. Using $\|P_{\leq N }\uu\|_{L^4_t L^\infty_{\lu}}^4 \leq \|P_{\leq N }\widehat\uu\|_{L^4_t L^1_{\lu}}^4 \leq C TN^2\|\uu\|_{L^\infty_t H^1_{\lu}}^4$, together with $\|P_{\leq N }u\|_{L^4_t L^{12}_{\lu}}^4\leq CT\|P_{\leq N }\uu\|_{L^\infty_t H_{\lu}^{\frac54}}^4\leq CTN \|\uu\|_{L^\infty_t H_{\lu}^{1}}^4$, and Sobolev embedding, this is bounded by 
\begin{eqnarray}  
	\|\Phi_{NL}^{(2)}(P_{\leq N}\uu,P_{>N}\uu)\|_{X^1_{\lu}} 
	\leq
	C T N^2  \|\uu\|_{L^\infty_t H^1_{\lu}}^5 
	\leq C_2 T N^2 A^5 \,.
\end{eqnarray}

To show that $\Phi$ is a contraction on $B_2,$ let $\vu \in B_2,$.
Then, similarly as above, one shows that
\begin{eqnarray}
	\|\Phi_{NL}^{(1)}(P_{\leq N}\uu,P_{>N}\uu) - 
	\Phi_{NL}^{(1)}(P_{\leq N}\vu,P_{>N}\vu)\|_{X^1_{\lu}} 
	&\leq& 
	C_1 A^3 \eta \|\uu-\vu\|_{X^1_{\lu}} \,.
\end{eqnarray}
Moreover, one obtains
\begin{eqnarray} 
	\|\Phi_{NL}^{(2)}(P_{\leq N}\uu,P_{>N}\uu)
	-\Phi_{NL}^{(2)}(P_{\leq N}\vu,P_{>N}\vu)\|_{X^1_{\lu}} 
	\leq C_1 T N^2 A^4  \|\uu-\vu\|_{X^1_{\lu}} \,.
\end{eqnarray}
Then, letting $0<T<\frac{\eta}{10 C_2 N^2 A^5}$, and choosing $\eta>0$ sufficiently small, it follows that $\Phi$ maps $B_2$ into itself, and is a strict contraction.  

While contraction mapping theorem gives a unique solution $\uu$ in $B_2,$ we must show that uniqueness holds in the larger space $X^1_{\lu}([0,T])\cap C^1_t H^1_{\lu}\left([0,T] \times \tor^3\right).$ Suppose that $\vu \in X^1_{\lu}([0,T])\cap C^1_t H^1_{\lu}\left([0,T] \times \tor^3\right)$ is a solution to the equation with $\vu(0) = \uu_0.$ There exists some $N' \geq 1$ such that $\| \vu\|_{X^1([0,T])} \leq 2\delta.$   If $N' > N,$ define a new ball $B'_2$ that contains both $\uu$ and $\vu$ and apply the contraction mapping argument to see that $\uu = \vu$ on a (possibly smaller) time interval $[0,T'].$ By repeating this argument, we achieve uniqueness in the larger space.
\endprf

\section{Stationary States of the NLSS on $\tor^3$}
\index{Stationary States of the NLSS on $\tor^3$%
@\emph{Stationary States of the NLSS on $\tor^3$}}%
\label{sec-10}


We now turn to the existence and nonlinear stability of stationary states of cubic and quintic NLS systems on three-dimensional flat tori. As in the previous chapters, the results hold for rectangular, rational and irrational tori. In this chapter we restrict ourselves to consider only the defocusing systems
\begin{equation}\label{DefNLSS}
\begin{cases}
& i\pd_t u_j = -\D u_j+ \r^{\a} u_j, \qquad j \in \mathbb{N}\\
&u_j(0,x) = u_{j,0}(x), \qquad x \in \tor^3,
\end{cases}
\end{equation}
where $\r:=\sum_{j \in \N} \l_j |u_j|^2$ for a given $\lu \in \ell^1_+,$ and $\a \in\{1,2\}.$ 

Stationary states $\{v_j\}_{j \in \N}$ are solutions to \eqref{DefNLSS} of the form
\begin{equation*}
v_j(t,x) = e^{-i\mu_j t}u_{j,0}(x)
\end{equation*}
where $\mu_j \in \R$ is the energy level of $u_{j,0}(x).$ As stated in the introductory chapter, the stationary states we find are minimizers of an energy-Casimir functional, which is the sum of the conserved energy and another function conserved by the flow of \eqref{DefNLSS}. 

In this chapter, we begin by defining Casimir-class functions and the stationary state equations corresponding to a Casimir-class function $f,$ then develop the definition and properties of the energy-Casimir functional $\sH_f$ determined by $f.$ Next, assuming the existence and uniqueness of the desired stationary states, we bound a nonlinear function of the distance between a stationary state and another solution to \eqref{DefNLSS} using the energy-Casimir functional. Finally, to prove the existence and uniqueness of the stationary states, we use the saddle point principle to find a dual functional to $\sH_f,$ for any Casimir-class $f,$ and use convexity theory to show that the dual functional has a unique maximizer. This maximizer corresponds to a stationary state of \eqref{DefNLSS} which minimizes $\sH_f.$

 
\section{Stationary states and Energy-Casimir Functionals}


 Define the state space for the NLSS as
\begin{align*}
\S=\Big\{ (\uu,\lu) \big\arrowvert\, &\uu = \{u_k\}_{k\in\mathbb{N}} \subset H^1(\tor^3) \text{ a complete orthonormal system in } L^2(\tor^3),\\
& \lu = \{\l_k\}_{k\in \mathbb{N}}\in \ell^1 \text{ with } \l_k \geq 0, \text{ and } \sum_{k=1}^{\infty} \l_k \| u_k \|_{H^1(\tor^3)}^2 < \infty \Big\}.
\end{align*}
In the previous chapters, we have shown that the defocusing cubic NLSS is globally well-posed in $\S,$ and for some $\eta > 0,$ the defocusing quintic NLSS is globally well-posed for initial data $(\uu_0, \lu) \in \S$ provided $\|\uu_0\|_{H^1_{\lu}} < \eta.$ 


\begin{defin}\label{def-Casimir-fct-1}
A function $f: \mathbb{R}\rightarrow\mathbb{R}$ is said to be of \textbf{Casimir class }  $\sC$ if it has the following properties:
\begin{enumerate}[label=(\roman*)]
\item $f$ is continuous, and there is $s_0 \in (0,\infty]$ such that $f(s) >0$ for $s \leq
 s_0$ and $f(s) =0$ for $s > s_0.$
 \item $f$ is strictly decreasing on $(-\infty, s_0]$ with $\lim_{s\rightarrow -\infty} f(s)=\infty.$
 \item there exist constants $\epsilon >0 $ and $C > 0$ such that 
 \begin{equation*}
 f(s) \leq C(1+s)^{(-5/2-\epsilon)}\text{ for }\quad s  \geq 0
\end{equation*}
\end{enumerate}
\end{defin}
An example of $f\in \sC$ with $s_0 = \infty$ is given by the Boltzmann distribution $f(s) = e^{-\beta s}$ for $\beta > 0.$


The stationary states that we seek are $(\uu_0,\lu_0) \in \S$ corresponding to a quadruple $(\uu_0,\lu_0,\muu_0, \r_0)$ with $\muu_0 = \{\mu_{0,k}\}_{k \in \mathbb{N}}\subset \R,$ and $\r_0 \in L^{\a +1}(\tor^3),$ such that for some $f \in \mathscr{C},$
\begin{equation}\label{StStates}
\left\{
\begin{aligned}
(-\D + \r_0^{\a}) u_{0,k} &=\mu_{0,k}u_{0,k}\; \text{ for all } k\in \mathbb{N}\\
\r_0 &=\sum_{k=1}^{\infty} \l_{0,k} | u_{0,k} |^2\\
\l_{0,k}&=f(\mu_{0,k}) \text{ for all } k \in \mathbb{N}
\end{aligned}
\right.
\end{equation}
where $\a = 1$ or $2$ throughout. The equation $\l_{0,k} = f(\mu_{0,k})$ demonstrates the role of the function $f \in \sC :$ a stationary state $u_{0,k}$ with energy $\mu_{0,k}$ has occupation probability $\l_{0,k}=f(\mu_{0,k}).$ We see that if $s_0$ is finite for $f \in \sC$, the NLSS is constrained to a finite number of occupied states. Thus we set $s_0 = \infty$ for the remainder of this chapter. 


The next proposition ensures that for any solution of the stationary state equations, $(\uu_0, \lu_0)$ is in the required state space, and $\r_0$ has the integrability required for the solution to have finite energy.
\begin{prop}\label{prop-fsum-conv-1}
Suppose the quadruple $(\uu_0,\lu_0,\muu_0, \r_0)$ satisfies the stationary state equations \eqref{StStates} with $f \in \sC,$ and  $\uu_0 = \{ u_{0,k}\}_{k=1}^{\infty}$ a complete orthonormal basis of $L^2(\tor^3).$ Then $\r_0 \in L^{\a + 1}(\tor^3),$ and $(\uu_0, \lu_0) \in \S.$
\end{prop}
\prf
First observe that the nonnegativity of $f$ immediately gives the nonnegativity of $\l_{0,k}$ for all $k,$ which implies $\r_0$ is nonnegative, thus $\mu_{0,k}$ is also nonnegative for all $k.$ From the first equation in \eqref{StStates}, we find
\begin{equation*}
\sum_{k=1}^{\infty}\l_{0,k}\int \left[|\n u_{0,k}|^2 + \r_0^{\a}|u_{0,k}|^2\right]\,dx = \sum _{k=1}^{\infty} \l_{0,k} \mu_{0,k}
\end{equation*}
The stationary state equations satisfied by $\r_0$ and $\l_{0,k}$ show that the previous equation may be rewritten in the form
\begin{equation}\label{SumConv}
\sum_{k=1}^{\infty} \left( \l_{0,k} \int | \n u_{0,k}|^2\,dx\right) + \int \r_0^{\a + 1}dx =  \sum _{k=1}^{\infty} f( \mu_{0,k}) \mu_{0,k}
\end{equation}

We claim the sum on the right hand side of \eqref{SumConv} is finite. Since $f \in \sC,$ 
\begin{equation*}
f(\mu_{0,k})\mu_{0,k} \leq C(1 + \mu_{0,k})^{-3/2 - \epsilon}
\end{equation*}
for each $k \in \N.$ Let $\{ \mu_{-\D,k} \}$ denote the complete set of eigenvalues of $-\D$ on $\tor^3,$ and observe that the nonnegativity of $\r_0$ implies $\mu_{0,k} \geq \mu_{-\D,k}.$ The estimate of Li and Yau \cite{LiYau} gives  $\mu_{-\D,k}\geq Ck^{2/3},$ where the constant $C$ depends only on the domain $\tor^3.$ Thus, for each $k \in \N,$ 
\begin{equation*}
f(\mu_{0,k})\mu_{0,k} \leq C(1 + \mu_{-\D,k})^{-3/2 - \epsilon} \leq Ck^{-1-\epsilon}
\end{equation*}
which proves our claim that the sum converges. 

As the right side of \eqref{SumConv} is finite, $\sum_{k=1}^{\infty} \l_{0,k} \int | \n u_{0,k}|^2\,dx$ is finite, and $\r_0 \in L^{\a + 1}(\tor^3).$ By the Poincar\'e inequality, $\sum_{k=1}^{\infty} \l_{0,k} \int | u_{0,k}|^2\,dx$ must also be finite, and we conclude $(\uu_0, \lu_0) \in \S.$
\endprf


\begin{rem}
In the cubic case, $\a=1,$ we have $\r \in L^2(\tor^3),$ and $\r$ serves as the potential function in the cubic NLSS. However, in the quintic case, $\a = 2,$  the potential function  is $\r^2.$ We have shown that $\r \in L^3(\tor^3),$ thus $\r^2 \in L^{3/2}(\tor^3).$ In order to generalize our arguments to apply to both the cubic and quintic NLSS, we will use potentials $V \in L_+^{\frac{\a + 1}{\a}}(\tor^3),$ which include the functions $\r^{\a}.$
\end{rem}


We now develop the energy-Casimir functional associated to a given Casimir-class function $f.$ For $f \in \mathscr{C}$, define
\begin{equation*}
F(s):= \int_s^\infty f(\sigma)\,d\sigma,\quad s \in \mathbb{R}.
\end{equation*}
$F$ is a nonnegative, continuously differentiable, decreasing function, strictly convex on its support. Furthermore, we have the bound
\begin{equation}\label{FRate}
F(s) \leq  C(1+s)^{(-3/2-\epsilon)} \text{ for }  s  \geq 0.
\end{equation}
The Legendre transform of F is given by
\begin{equation}\label{eq-Leg-def-1}
F^*(\l) = \sup_{s \in \mathbb{R}}\big(\l s - F(s) \big) \quad \l \in \mathbb{R}.
\end{equation}
Since $F$ is differentiable with $F'=-f$, $F^*$ is differentiable, and $(F^*)' =(-f)^{-1}$. In particular, the supremum in \eqref{eq-Leg-def-1} is attained at $s=f^{-1}(-\lambda)$, and the Legendre transform of $F$ is given by 
\begin{equation}\label{eq-Legendre-id-1}
	F^*(-\lambda)=\lambda \mu-F(\mu).
\end{equation}
where $\mu=f^{-1}(\lambda)$.
Moreover, the Legendre transform is an involution, $F^{**}=F$.

We recall that the energy of a solution $\uu$ to the defocusing NLSS \eqref{DefNLSS} determined by $\lu$ is defined as
\begin{equation*}
E_{\lu}(\uu) :=\frac{1}{2}\sum_k \l_k \|\n u_k\|_{L^2(\tor^3)}^2 + \frac{1}{2(\a +1)}\int_{\tor^3} \r^{\a+1}\,dx 
\end{equation*}
and is conserved by the flow of the system.
\begin{defin}
Let $(\uu,\lu) \in \S.$ For a fixed $f \in \sC,$ we define the \emph{energy-Casimir functional } determined by $f$ as
\begin{align}\label{ECF}
\sH_f(\uu, \lu) &:= \sum_{k}F^*(-\l_k) + 2E_{\lu}(\uu) \\
								&=\sum_{k}\big(F^*(-\l_k) + \l_k \int_{\tor^3}|\n u_k |^2\,dx\big) +\frac{1}{\a +1}\int_{\tor^3}\r_{\uu}^{\a +1}\,dx
\end{align}
Since $\lu$ and $E_{\lu}(\uu)$ are constant in time, $\sH_f$ is also a conserved quantity of the defocusing NLSS. We will prove the stability of stationary solutions of the NLSS employing $\sH_f$ in a similar way as Lyapunov functions are used for the corresponding problem in classical Hamiltonian dynamics. This approach is often referred to as the energy-Casimir method.

We remark that the convergence of $\sum_{k}F^*(-\l_k)$ follows from 
\begin{equation}
	\sum_k F^*(-\l_{k}) = \sum_k \l_{k}\mu_{k} -\sum_k F(\mu_{k}) 
\end{equation}
where $\lambda_k=f(\mu_k)$; see \eqref{eq-Legendre-id-1} and \eqref{eq-Fmu-sum-1}, below. The convergence of $\sum_k\l_k\mu_k $ is proven in Proposition \ref{prop-fsum-conv-1}, and that of $\sum_k F(\mu_k)$ in Lemma \ref{Lemma1}.

\end{defin}


We conclude this section with some useful properties of $f$ and $F$ for $f \in \sC.$
\begin{lemma}\label{Lemma1}
Let $f\in \sC.$
\begin{enumerate}[label=(\roman*)]
\item For every $\beta > 1,$ there exists $C = C(\beta) \in \mathbb{R}$ such that 
\begin{equation*} 
F(s) \geq -\beta s + C, \quad s\leq 0
\end{equation*}
\item If $V \in L_+^{\frac{\a + 1}{\a}}(\tor^3)$  then both $f(-\D + V)$ and $F(-\D + V)$ are trace class.
\end{enumerate}
\end{lemma}

\prf
Part (i) of the lemma follows directly from the properties of $F(s).$ In particular, as $F(s)$  strictly convex on its support, its graph lies above any of its tangent lines. Since $F(s)$ is decreasing, with $\lim_{s \rightarrow -\infty}F(s) = \infty,$ for any $\beta > 1,$ there is some $s<0$ such that the tangent line to $F$ at $s$ has slope $-\beta.$ 

To prove part (ii) of the lemma, let $\{\mu_{V,k}\}_{k=1}^{\infty}$ be the complete set of eigenvalues of $-\D + V,$  and let  $\{\mu_{-\D,k}\}_{k=1}^{\infty}$ be the complete set of eigenvalues of $-\D$ on $L^2(\tor^3).$  As $V$ is nonnegative, we have $\mu_{V,k} \geq \mu_{-\D,k} \geq C k^{\frac{2}{3}}$ for each $k \in \N.$ As $F(s)$ decreases faster than $(1 + s)^{-3/2}$ for $s \geq 0,$ we find  
\begin{align*} 
\text{Tr} F(-\D + V) &= \sum_{k=1}^{\infty} F(\mu_{V,k}) \leq \sum_{k=1}^{\infty}F(\mu_{-\D,k})\\
& \leq \sum_{k=1}^{\infty}F\big( C k^{\frac{2}{3}}\big) \leq C \sum_{k=1}^{\infty} \big(1 + k^{\frac{2}{3}}\big)^{-\frac{3}{2} - \epsilon}.
\end{align*}
The last series is convergent, thus $F(-\D + V)$ is trace class. As $f(s)$ is nonnegative and decreases at a rate faster than $F(s),$  $f(-\D + V)$ is also trace class.
\endprf

\begin{lemma}\label{FConvex}
For $\phi \in H^1(\tor^3)$ with $\| \phi \|_{L^2(\tor^3)}= 1,$ and $V \in L^{\frac{\a + 1}{a}}_+(\tor^3),$ 
\begin{equation*}
F \left( \la \phi , (-\D + V)\phi \ra \right) \leq \la \phi, F(-\D + V)\phi \ra
\end{equation*}
with equality if $\phi$ is an eigenstate of $-\D + V.$
\end{lemma}

\prf
Using the spectral theorem, let $P_{\g}$ be the family of orthogonal projections onto the eigenspaces of $-\D + V,$ and write
\begin{equation*}
-\D + V = \int_0^{\infty}\g dP_{\g}
\end{equation*}
 For any $\phi$ satisfying the hypotheses, the spectral measure of $-\D + V$ with respect to $\phi$ is given by
\begin{equation*}
\la \phi , dP_{\g}\phi\ra =:d\nu(\g),
\end{equation*}
 which is indeed a probability measure. Since F is convex, we apply Jensen's inequality to conclude
\begin{equation*}
F\left(\int_0^{\infty}\g d \nu(\g)\right) \leq \int_0^{\infty} F(\g)d\nu(\g)
\end{equation*}
which is equivalent to the inequality in the lemma. 

If $\phi$ is an eigenstate of $-\D + V,$ with eigenvalue $\g_0$ then $d\nu(\g)$ is a Dirac measure at $\g_0,$ and each side of the above inequality is $F(\g_0).$ 
\endprf
\begin{cor}
For $\phi \in H^1(\tor^3)$ with $\| \phi \|_{L^2(\tor^3)}= 1,$ $V \in L^{\frac{\a + 1}{a}}_+(\tor^3),$ and fixed but arbitrary $\s \in \R$
\begin{equation*}
F \left( \la \phi , (-\D + V + \s)\phi \ra \right) \leq \la \phi, F(-\D + V + \s)\phi \ra
\end{equation*}
with equality if $\phi$ is an eigenstate of $-\D + V.$
\end{cor}

\prf
Note that $f_{\s}(s):= f(s + \s) \in \sC$ for $f(s) \in \sC,$ since we may take the cutoff $s_0$ as large as we wish. The corollary follows by applying the previous lemma to $F_{\s}(s) = F(s + \s).$
\endprf


\section{Nonlinear Stability of Stationary States}

For a given $f \in \sC,$ we define the functional $\Pf(\uu,\lu, V),$ as follows, where $(\uu,\lu) \in \S,$ and $V \in L_+^{\frac{\a + 1}{\a}}(\tor^3),$ with $\a =1$ (cubic) or $\a=2$ (quintic). 
\begin{equation}\label{Pf}
\Pf(\uu,\lu, V):=\sum_{k=1}^{\infty}\Big[F^*(-\l_k)+\l_k \int\big(|\n u_k|^2 + V |u_k|^2\big)\,dx \Big]
\end{equation}
\begin{rem}\label{PsiEnCas}
Note that if $(\uu,\lu)$ is a solution of the NLSS \eqref{NLSS} with corresponding density function $\r \in L_+^{\a +1}(\tor^3),$
\begin{align}
\Pf(\uu,\lu, \r^\alpha):&=\sum_{k=1}^{\infty}\Big[F^*(-\l_k)+\l_k \int |\n u_k|^2\,dx\Big] + \int \r^{\a +1}\,dx\notag\\
&=\sH_f(\uu,\lu) + \frac{\a}{\a+1}\int \r^{\a + 1} \, dx
\end{align}
\end{rem}


\begin{lemma}\label{TrBound}
Let $V \in L_+^{\frac{\a + 1}{\a}}(\tor^3),$ with $\a =1$ or $\a=2.$ For any $(\uu, \lu) \in \S,$ and any $f \in \sC,$
\begin{equation}\label{TrBoundIneq}
 \Pf(\uu,\lu, V) \geq -\mathrm{Tr}[F(-\D + V)],
\end{equation}
with equality if $(\uu, \lu) = (\uu_V, \lu_V),$ where $\uu_V = \{u_{V,k} \}$ is an orthonormal sequence of eigenfunctions of $-\D + V$ with eigenvalues $\muu_V= \{ \mu_{V,k} \},$ satisfying $\l_{V,k}= f(\mu_{V,k})$ for all $k \in \N.$
\end{lemma}

\prf
Set 
\begin{equation*}
\mu_k=  \la u_k, (-\D + V) u_k \ra =\int \left(| \n u_k|^2 + V | u_k|^2\right)\,dx
\end{equation*} 
By the inequality $F^*(-\l) + \l \mu \geq -F(\mu)$ for $\l, \mu \in \mathbb{R},$ and Lemma \ref{FConvex}, we have
\begin{align*}
\sum_{k=1}^{\infty}\left[ F^*(-\l_k) + \l_k \int \left(| \n u_k|^2 + V | u_k|^2\right)\,dx \right] 
 &\geq -\sum_{k=1}^{\infty} F\left( \la u_k, (-\D +V) u_k \ra\right)\\
&\geq - \sum_{k=1}^{\infty} \la u_k, F(-\D +V) u_k \ra\\
&=-\mathrm{Tr}[ F(-\D +V)]
\end{align*}

Now suppose $\uu_{V} =  \{u_{V,k} \}$ is the orthonormal sequence of eigenfunctions of $-\D + V$ with eigenvalues $\muu_V= \{ \mu_{V,k} \},$ so that  $\mu_{V,k} = \la u_{V,k}, ( -\D + V)u_{V,k}\ra.$ 
\begin{equation*}
\text{Tr}[F(-\D +V)] = \sum_{k=1}^{\infty} F(\mu_{V,k})
\end{equation*}
By definition, $\l_{V,k} =f(\mu_{V,k}) = -F'(\mu_{V,k}),$ for each $k \in \N.$ By the conjugate relationship  
$\mu_{V,k} = (F^*)'(-\l_{V,k}),$ so that $F^*(-\l_{V,k}) = -\l_{V,k}\mu_{V,k} - F(\mu_{V,k}).$ Summing on $k$ gives
\begin{equation}\label{eq-Fmu-sum-1}
-\sum_{k=1}^{\infty} F(\mu_{V,k}) =  \sum_{k=1}^{\infty}\big[ F^*(-\l_{V,k}) + \l_{V,k}\mu_{V,k}\big],
\end{equation}
which is precisely the statement of equality in \eqref{TrBoundIneq}.
\endprf

\begin{cor}\label{TrBoundCor}
Let $V \in L_+^{\frac{\a + 1}{\a}}(\tor^3),$ with $\a =1$ or $\a=2.$ For any $\s \in \R,$ any  $(\uu, \lu) \in \S,$ and any $f \in \sC,$
\begin{equation}\label{sigTraceBound}
 \Pf(\uu,\lu, V) + \s\sum_{k=1}^{\infty} \l_k \geq -\mathrm{Tr}[F(-\D + V + \s)], 
\end{equation}
with equality if $(\uu, \lu) = (\uu_V, \lu_V),$ where $\uu_V = \{u_{V,k} \}$ is an orthonormal sequence of eigenfunctions of $-\D + V$ with eigenvalues $\muu_V= \{ \mu_{V,k} \},$ with $\l_{V,k}= f(\mu_{V,k} + \s)$ for all $k \in \N.$
\end{cor}
\begin{proof}
We use the same argument of previous lemma, replacing $\mu_k$ with $\mu_k + \s$ throughout. 
\end{proof}


\begin{theorem}[Nonlinear Stability of Stationary States]\label{StabThm}
Let $(\uu_0,\lu_0,\muu_0, \r_0)$ be a stationary state of the defocusing NLSS \eqref{DefNLSS} with $\a \in \{1,2\}.$ Suppose $(\uu_0, \lu_0) \in \mathscr{S},$ and $\l_{0,k}=f(\mu_{0,k})$ for some $f \in \mathscr{C}$ and all $k \in \mathbb{N}.$ Let $\sH_f$ be the energy-Casimir functional determined by $f.$ If $(\uu(t), \lu)$ is another solution of the defocusing NLSS on the time interval $[0,T)$ with initial datum $(\uu(0), \lu) \in \mathscr{S},$ then 
\begin{equation*}
\frac{1}{\a+1} \| \r_{\uu(t), \lu} - \r_0\|^{\a+1}_{L^{\a+1}(\tor^3)} \leq 
|\sH_f(\uu(0),\lu) - \sH_f(\uu_0,\lu_0)|, \quad t \geq 0.
\end{equation*}
for all $t \in [0, T).$
\end{theorem}

\begin{proof}
\emph{ Cubic Case.} Let $(\uu_0, \lu_0, \muu_0, \r_0)$ be a stationary state  of the cubic NLSS with $f \in \sC$ satisfying $f(\mu_{0,k}) = \l_{0,k},$ for all $k\in \N$. Suppose $(\uu,\lu)$ is a solution of equation \eqref{DefNLSS} for $\a=1,$ with initial datum $(\uu(0), \lu)\in \S,$ and let $\r \in L^2_+(\tor^3)$ be the particle density corresponding to $(\uu,\lu).$ We have 
\begin{align}
\frac{1}{2}&\| \r - \r_0 \|_2^2 = \frac{1}{2}\int \big(\r^2 -2\r \r_0 +   \r_0^2\big)dx\notag\\
&=\sH_f(\uu,\lu) - \sum_{k=1}^{\infty}\Big[F^*(-\l_k)+\l_k \int\big(|\n u_k|^2 + \r_0 |u_k|^2\big)\,dx \Big] + \frac{1}{2} \int  \r_0^2\,dx \notag\\ 
& \leq\sH_f(\uu,\lu) + \text{Tr}[F(-\D + \r_0)]  + \frac{1}{2} \int  \r_0^2\,dx \label{TrIneq}\\
&=\sH_f(\uu,\lu) - \Pf(\uu_0,\lu_0,\r_0) + \frac{1}{2} \int  \r_0^2\,dx\label{TrEq}\\
&=\sH_f(\uu(0),\lu) - \sH_f(\uu_0,\lu_0),\notag
\end{align}
where we have used Lemma \ref{TrBound} to establish \eqref{TrIneq} and \eqref{TrEq}.

\emph{ Quintic Case.}
Let $(\uu_0, \lu_0, \muu_0, \r_0)$ be a stationary state of the quintic NLSS with $f \in \sC$ satisfying $f(\mu_{0,k}) = \l_{0,k},$ for all $k\in \N$. Suppose $(\uu,\lu)$ is a solution of \eqref{DefNLSS} for $\a =2$ on time interval $[0,T)$ with initial datum $(\uu(0), \lu)\in \S,$ and let $\r \in L^3_+(\tor^3)$ be the particle density corresponding to $(\uu,\lu).$

First, note that since $\r, \r_0 \geq 0$ on $\tor^3,$ we have
\begin{equation}\label{QStabIneq1}
\int | \r - \r_0 |^3\,dx \leq \int(\r - \r_0)^2(\r + \r_0)\, dx = \int (\r^3 - \r^2\r_0 - \r \r_0^2 + \r_0^3)\,dx 
\end{equation}
By the geometric-arithmetic means inequality, 
\begin{equation*}
\r_0^2\r = \r_0(\r_0 \r) \leq \frac{1}{2}\r_0 ( \r_0^2 + \r^2 ),
\end{equation*}
from which we obtain  
\begin{equation}\label{rhosIneq}
-\r^2\r_0 \leq -2\r_0^2\r + \r_0^3.
\end{equation}
Using estimates \eqref{QStabIneq1} and \eqref{rhosIneq}, we find
\begin{equation}\label{QStabIneq2}
\int | \r - \r_0 |^3\,dx  \leq \int \big(\r^3 - 3\r_0^2\r + 2\r_0^3\big)\,dx
\end{equation}
Proceeding as we did in the proof of the cubic case:
\begin{align}
\frac{1}{3}&\| \r - \r_0 \|_{L^3(\tor^3)}^3 \leq \frac{1}{3}\int \big(\r^3 - 3\r_0^2\r + 2\r_0^3\big)\,dx\notag\\
&=\sH_f(\uu,\lu) - \sum_{k=1}^{\infty}\big[F^*(-\l_k)+\l_k \int\big(|\n u_k|^2+\r_0^2|u_k|^2\big)\,dx \big] + \frac{2}{3} \int  \r_0^3\,dx\notag\\
&=\sH_f(\uu,\lu) - \Pf(\uu,\lu,\r_0^2) + \frac{2}{3} \int  \r_0^3\,dx\notag\\
& \leq \sH_f(\uu,\lu)  + \text{Tr}[F(-\D + \r^2_0)] + \frac{2}{3} \int  \r_0^3\,dx\label{QTrIneq}\\
&=\sH_f(\uu,\lu) - \Pf(\uu_0,\lu_0,\r_0^2) + \frac{2}{3} \int  \r_0^3\,dx\label{QTrEq}\\
&=\sH_f(\uu(0),\lu) - \sH_f(\uu_0,\lu_0)\notag,
\end{align}
where we have used Lemma \ref{TrBound} to establish \eqref{QTrIneq} and \eqref{QTrEq}.
\end{proof}


\section{Deriving the Dual Functional}


We now turn to the problem of existence of stationary states satisfying \eqref{StStates}. For each $f \in \sC,$ we define a dual functional to $\sH_f(\uu,\lu).$ First, for fixed $f \in \sC$ and fixed $\L > 0,$ we use the saddle point principle to define:
\begin{multline*}
\G(\uu, \lu, V, \s) := \sum_{k=1}^{\infty}\left[F^*(-\l_k) + \l_k \int \big(|\n u_k|^2 + V|u_k|^2\big)dx\right] \\
 - \frac{\a}{\a + 1}\int V^{\frac{\a+1}{\a}}\,dx + \s\left[\sum_{k=1}^{\infty}\l_k - \L\right]
\end{multline*}
where $\uu = \{u_k\}$ is an orthonormal basis of $L^2(\tor^3),$ $\lu = \{\l_k\} \in \ell^1_+,$ and $V \in L_+^{\frac{\a  + 1}{\a}}.$ The variable $\s \in \R$ plays the role of a Lagrange multiplier.


The following lemma illustrates the relationship between the functional $\G$ and the energy-Casimir functional.
\begin{lemma}
For any $\uu, \lu, \s$ we have 
\begin{equation}
\sup_V \G(\uu, \lu, V, \s) = \sH_f (\uu, \lu) + \s\left[\sum_{k=1}^{\infty} \l_k - \L\right],
\end{equation}
and the supremum occurs when $V = \big( \sum_{k=1}^{\infty} \l_k |u_k|^2 \big)^{\a}$
\end{lemma}
\begin{proof}
For arbitrary $(\uu, \lu) \in \S,$ let $\r_{\uu,\lu} = \sum_k \l_k |u_k|^2.$ Suppose $\a = 1.$ We may write $\G(\uu,\lu,V,\s)$ equivalently as 
\begin{align*}
\G(\uu,\lu,V,\s)& = \sum_{k=1}^{\infty}\left[F^*(-\l_k) + \l_k \int |\n u_k|^2 dx\right]  + \frac{1}{2} \int \r_{\uu,\lu}^2 dx - \frac{1}{2} \int \r_{\uu,\lu}^2 dx \\
&\phantom{ = \sum_k\big[F^*(-\l_k) + \l_k }+ \int V \r_{\uu,\lu}\, dx - \frac{1}{2}\int V^2\,dx + \s\left[\sum_{k=1}^{\infty} \l_k - \L\right] \\
&= \sH_f(\uu,\lu) - \frac{1}{2}\| \r_{\uu,\lu} - V \|_{L^2(\tor^3)}^2 +  \s\left[\sum_{k=1}^{\infty} \l_k - \L\right].
\end{align*}
Clearly, $\G(\uu,\lu,V,\s)$ has a maximum for $V = \r_{\uu,\lu},$ and we have the desired supremum.

Now, consider the case $\a = 2.$ This time, we write $\G(\uu,\lu,V,\s)$ in the equivalent form
\begin{align*}
\G(\uu,\lu,V,\s)& = \sum_{k=1}^{\infty}\left[F^*(-\l_k) + \l_k \int |\n u_k|^2 dx\right]  + \frac{1}{3} \int \r_{\uu,\lu}^3 dx - \frac{1}{3} \int \r_{\uu,\lu}^3 dx \\
&\phantom{ = \sum_k\big[F^*(-\l_k) }+ \int V \r_{\uu,\lu}\, dx - \frac{2}{3}\int V^{3/2}\,dx + \s\left[\sum_{k=1}^{\infty} \l_k - \L\right] \\
&= \sH_f(\uu,\lu) - \frac{1}{3}\int \big(\r_{\uu,\lu}^3 - 3\r_{\uu,\lu}V + 2V^{3/2}\big) dx+  \s\left[\sum_{k=1}^{\infty} \l_k - \L\right].
\end{align*}
Since $V \in L^{3/2}_+,$ it has a nonnegative square root. Let $\r := \sqrt{V} \in L^3_+,$ so that
\begin{align*}
\G(\uu,\lu,\r^2,\s) & = \sH_f(\uu,\lu) - \frac{1}{3}\int \big(\r_{\uu,\lu}^3 - 3\r^2\r_{\uu,\lu} + 2\r^3\big) dx+  \s\left[\sum_{k=1}^{\infty}\l_k - \L\right]\\
&= \sH_f(\uu,\lu) - \frac{1}{3}\int (\r_{\uu,\lu} + 2\r)(\r_{\uu,\lu} - \r)^2 dx+  \s\left[\sum_{k=1}^{\infty} \l_k - \L\right].
\end{align*} 
As $\r_{\uu,\lu}$ and $\r$ are nonnegative, we have $- \frac{1}{3}\int (\r_{\uu,\lu} + 2\r)(\r_{\uu,\lu} - \r)^2 dx \leq 0$ with equality precisely when $\sqrt{V} =  \r = \r_{\uu,\lu} ,$ which proves the lemma in the case $\a = 2.$
\end{proof}


Let us now derive a useful representation of the dual function defined by $\Phi(V,\s):=\inf_{\uu, \lu}\G(\uu,\lu,V\s).$ First, note that $\G$ can be written in the form
\begin{equation*}
\G(\uu, \lu, V, \s) = \Pf(\uu, \lu, V) - \frac{\a}{\a+1}\int V^{\frac{\a + 1}{\a}}dx + \s\left[\sum_{k=1}^{\infty}\l_k - \L\right].
\end{equation*}
By Lemma \ref{TrBound}, we have the lower bound
\begin{equation*}
\G(\uu, \lu, V, \s) \geq -\text{Tr}\big[F(-\D + V + \s)\big] - \s \L - \frac{\a}{\a+1}\int V^{\frac{\a + 1}{\a}}dx
\end{equation*}
with equality if $(\uu, \lu) = (\uu_V, \lu_V),$ where $\uu_V$ is the complete set of eigenstates of $-\D + V$ with corresponding eigenvalues $\muu_V$ such that $\l_{V,k} = f(\mu_{V,k} + \s)$ for all $k \in \N .$ Therefore, we have the following expression for $\Phi$
\begin{equation*}
\Phi(V,\s):=  - \frac{\a}{\a+1}\int V^{\frac{\a + 1}{\a}}dx -\text{Tr}\big[F(-\D + V + \s)\big] - \s \L. 
\end{equation*}


\section{Existence and Uniqueness of Stationary States}
\label{sec-14}

\begin{theorem}\label{PhiMax}
The functional $\Phi(V,\s)$, determined by a given $\L >0$ and $f \in \sC$, is continuous, strictly concave, bounded from above, and $-\Phi(V,\s)$ is coercive. Thus $\Phi(V,\s)$ has a unique maximizer $(V_0, \s_0).$ This maximizer uniquely determines a stationary state $(\uu_0, \lu_0,\muu_0,\r_0)$ as follows: $\uu_0 = \{u_{0,k}\}$ is the set of orthonormal eigenstates  of $-\D + V_0$ with corresponding eigenvalues $\muu_0 = \{\mu_{0,k}\},$   $ \l_{0,k}:=f(\mu_{0,k} + \s_0)$ for $k \in \N$ satisfies $\sum_{k=1}^{\infty} \l_{0,k} = \L ,$  and $V_0 = \r_0^\a,$ where $\r_0 := \sum_{k=1}^{\infty} \l_{0,k}|u_{0,k}|^2.$ 
\end{theorem}

\begin{proof}
 For notational convenience, let $q:= \frac{\a + 1}{\a},$ so that $q = 2$ corresponds to the cubic NLSS while $q = \frac{3}{2}$ corresponds to the quintic NLSS. For fixed but arbitrary $f \in \sC$ and $\L > 0,$ $\Phi : L^q_+(\tor^3) \times \R \rightarrow \R$ is given by
\begin{equation*}
\Phi(V,\s)= - \frac{1}{q}\int V^q\,dx -\text{Tr}\big[F(-\D + V + \s)\big] - \s \L .
\end{equation*}


\emph{ $\Phi$ is strictly concave.} We begin by proving that $\text{Tr}\big[F(-\D + V + \s)\big]$ is convex. To this end, suppose $(V_j, \s_j) \in L^q(\tor^3) \times \R$ for $j = 1, 2,$ and consider the expression
\begin{equation*}
F\left(\big\la \psi, [ r(-\D + V_1 + \s_1) + (1-r)(-\D + V_2 + \s_2)]\psi \big\ra \right),
\end{equation*}
where $\psi  \in H^1(\tor^3)$  with $\| \psi \|_{L^2(\tor^3)} = 1,$ and  and $0 < r < 1$
By Lemma \ref{FConvex} and convexity of $F,$ we have 
\begin{multline*}
F\left(\big\la \psi, [ r(-\D + V_1 + \s_1) + (1-r)(-\D + V_2 + \s_2)]\psi \big\ra \right) \\
\leq r \big\la \psi, F(-\D + V_1 + \s_1) \psi \big\ra + (1-r) \big\la \psi, F(-\D + V_2 + \s_2) \psi \big\ra
\end{multline*}

Now, let $\{\psi_k\}$ be the complete set of eigenstates of  
\[r(-\D + V_1 + \s_1) + (1-r)(-\D + V_2 + \s_2).\] 
Using the definition of trace and the previous inequality, we have
\begin{multline} \label{TrConv}
\sum_{k=1}^{\infty} F \left(\big\la \psi_k, [ r(-\D + V_1 + \s_1) + (1-r)(-\D + V_2 + \s_2)]\psi_k \big\ra \right)\\
\leq r \sum_{k=1}^{\infty}\big\la \psi_k, F(-\D + V_1 + \s_1) \psi_k \big\ra \\+ (1 - r)\sum_{k=1}^{\infty}\big\la \psi_k, F(-\D + V_2 + \s_2) \psi_k \big\ra .
\end{multline}
Thus, $\text{Tr}\big[F(-\D + V + \s)\big]$ is convex.

As the remaining two terms $ - \frac{1}{q}\int V^q\,dx -\s \L $ are clearly concave in $(V,\s),$ we conclude $\Phi(V, \s)$ is concave. To show that $\Phi(V,\s)$ is strictly concave, suppose equality holds in the concavity statement for $\Phi(V,\s).$ This reduces to the equality of the expressions
\begin{multline}\label{strictConvTr}
\text{Tr}\big[F\big(r(-\D + V_1 + \s_1) + (1-r)(-\D + V_2 + \s_2)\big)\big] \\
- r\text{Tr}\big[F(-\D + V_1 + \s_1)] - (1-r)\text{Tr}\big[F(-\D + V_2 + \s_2)\big]
\end{multline}
and
\begin{equation}\label{strictConvInt}
 \frac{1}{q}\left(r\int V_1^q dx + (1-r)\int V_2^q dx - \int(rV_1 + (1-r)V_2)^q dx\right). 
\end{equation}
By convexity of $\text{Tr}\big[F(-\D + V + \s)\big],$ the expression \eqref{strictConvTr} is nonpositive. On the other hand, by the convexity of $\int V^q\,dx$ for $q=2$ and for $q=\frac{3}{2},$ the expression \eqref{strictConvInt} is nonnegative. For equality to hold, both expressions must equal zero. 

As $\int V^q dx$ is indeed strictly convex on the domain $V \in L^q_+$ for $q = 2$ and $q = \frac{3}{2},$ setting \eqref{strictConvInt} equal zero yields $V_1 = V_2.$ Next we set the expression \eqref{strictConvTr} equal to zero, which is equivalent to the case of equality in \eqref{TrConv}. In this case, the strict convexity of $F$ implies that for all $k \in \N,$
\begin{equation*}
\la\psi_k, F(-\D + V_1 + \s_1) \psi_k \ra = \la\psi_k, F(-\D + V_2 + \s_2) \psi_k \ra
\end{equation*}

Thus, the operators $-\D + V_1 + \s_1$ and $-\D + V_2 + \s_2$ have the same set of eigenvectors $\{\psi_k\}$ with the same eigenvalues. We combine this with the previous requirement that $V_1 = V_2$ to see that we must have $\s_1 = \s_2,$ which proves the strict concavity of $\Phi(V, \s).$


\emph{$\Phi$ is bounded from above, and $-\Phi$ is coercive.} Note that since $\L >0,$ we must distinguish the cases $\s \geq 0$ and $\s <0.$ First, suppose $\s$ is nonnegative. As $F$ is a positive, decreasing function, we immediately find
\begin{equation}\label{PhiCoerPos}
\Phi(V,\s) \leq  -\frac{1}{q}\|V\|_{L^q(\tor^3)}^q  - \s \L \leq 0
\end{equation}

Now consider the case  $\s < 0.$ Let $\mu_{V,1}$ be the ground state energy of $-\D + V.$ Again using positivity of $F,$ we obtain the upper bound
\begin{equation}\label{PhiUpBd2}
\Phi(V,\s) \leq -\frac{1}{q}\|V\|_{L^q(\tor^3)}^q-F(\mu_{V,1}+\s) - \s \L 
\end{equation}
By definition, 
\begin{equation*}
\mu_{V,1} := \inf_{\psi} \int |\n \psi |^2 + V|\psi|^2 dx
\end{equation*}
where the infimum is taken over all $\psi \in H^1(\tor^3)$ satisfying $\| \psi \|_{L^2(\tor^3)} = 1.$ Choose $\psi= \big(\text{vol }\tor^3\big)^{-\frac{1}{2}}$ in the above integral.  We have
\begin{align*}
\mu_{V,1}& \leq \int |\n \psi |^2 + V|\psi|^2 dx =\int V\psi^2  \,dx\\
&\leq \| V\|_{L^q(\tor^3)} \| \psi^2\|_{L^{q'}(\tor^3)} = C_1 \| V\|_{L^q(\tor^3)},
\end{align*}
where $q'$ is the H\"older conjugate of $q,$ and  $C_1:=\big(\text{vol }\tor^3\big)^{-\frac{1}{q}}.$ For $\s$ satisfying $\s \leq -C_1 \|V\|_{L^q} \leq -\mu_{V,1},$ part (i) of  Lemma \ref{Lemma1} guarantees that for any $\beta > 1$ there is some constant $C_2$ such that $-F(\mu_{V,1} + \s) \leq \beta(\mu_{V,1} + \s) - C_2.$ Choosing $\beta > \max\{ \L, 1 \}$ yields 
\begin{align*}
\Phi(V,\s) & -\frac{1}{q}\|V\|_{L^q(\tor^3)}^q-F(\mu_{V,1}+\s) - \s \L \\
&\leq  -\frac{1}{q}\|V\|_{L^q(\tor^3)}^q + \beta\mu_{V,1} +(\beta - \L)\s - C_2 \\
& \leq -\frac{1}{q}\|V\|_{L^q(\tor^3)}^q + \beta C_1 \|V \|_{L^2 }+(\beta - \L)\s   - C_2
\end{align*}
By the last inequality above and elementary calculus, there exists a positive constant $C_3$ such that 
\begin{equation}\label{PhiCoerNeg}
\Phi(V,\s) \leq -\frac{1}{2q}\|V\|_{L^q(\tor^3)}^q  + C_3 + (\beta - \L)\s - C_2
\end{equation}
on the interval $\s \leq -C_1 \|V\|_{L^q}.$ The inequalities \eqref{PhiCoerPos} and \eqref{PhiCoerNeg} together show that $\Phi(V, \s)$ is bounded above and that $-\Phi(V, \s)$ is coercive.


\emph{$\Phi$ is continuous.} The continuity of $\Phi$ is clear for all but the trace term. As it is convex, Proposition 2.5 in Chapter 1 of \cite{EkTem} implies the trace term is continuous on its support, provided it is proper and bounded above by a constant on some open set. The trace term is  proper as $F$ is nonnegative and trace class, and the local upper bound follows from the fact that $F$ is decreasing. Indeed, for any fixed $\s_0 \in \R,$ $\text{Tr}\big[F(-\D + V + \s)\big]$ is bounded above by $\text{Tr}\big[F(-\D + \s_0)\big] < \infty$ on the interval $\s > \s_0.$ 


\emph{$\Phi$ has a unique maximizer, corresponding to a stationary state.} By standard convexity theory, $\Phi$ has a unique maximum, occurring at some $(V_0, \s_0).$ Let $\uu_0 = \{u_{0,k}\}$ denote the complete set of orthonormal eigenfunctions of $-\D + V_0,$ with corresponding eigenvalues $\muu_0 = \{\mu_{0,k}\},$ and let $\l_{0,k} = f(\mu_{0,k} + \s_0).$ As $\s_0$ is a critical point for $\Phi(V_0, \s)$ and $F' = -f,$ we find
\begin{align*}
0 &=\left. \frac{d\Phi(V_0, \s)}{ds}\right|_{\s = \s_0} = \text{Tr}\big[f(-\D + V_0 + \s_0)\big]  - \L\\
&=\sum_{k=1}^{\infty} f(\mu_{0,k} + \s_0) - \L = \sum_{k=1}^{\infty} \l_{0,k} - \L
\end{align*} 
Thus $\sum_k \l_{0,k} = \L,$ as claimed. 

As $V_0$ is the maximizer of $\Phi(V,\s_0),$ it satisfies the Euler-Lagrange equation
\begin{equation*}
-V_0^{q-1} + \sum_{k=1}^{\infty} f(\mu_{0,k} + \s_0) |u_{0,k}|^2 = 0.
\end{equation*}
Note that since $q - 1 = \frac{1}{\a},$ the equation above gives $V_0^{\frac{1}{\a}} = \sum_{k=1}^{\infty} \l_{0,k} |u_{0,k}|^2 = \r_0.$ This concludes the proof for the existence and uniqueness of stationary states.
\end{proof}


\begin{prop}
Let the hypotheses of Theorem \ref{PhiMax} be satisfied. Suppose $(\uu_0,\lu_0,\muu_0, \r_0)$ is the unique stationary state corresponding to the maximizer $(V_0, \s_0)$ of the functional $\Phi$ determined by $f$ and $\L.$ Then $\Phi(V_0, \s_0) = \sH_f(\uu_0, \lu_0).$
\end{prop}
\prf
At $(V_0, \s_0),$ we have $V_0 = \r_0^{\a},$ and $\sum_{k =1}^{\infty} \l_{0,k} = \L.$ Using Corollary \ref{TrBoundCor} and Remark \eqref{PsiEnCas}, we find
\begin{align*}
\Phi(V_0, \s_0) &= -\text{Tr}\big[F(-\D + \r_0^{\a} + \s_0)\big] - \s_0 \L - \frac{\a}{\a+1}\int \r_0^{\a + 1}dx\\
&= \Pf(\uu_0,\lu_0, \r_0^{\a}+\s_0) -\s_0\L -\frac{\a}{\a+1}\int \r_0^{\a + 1}dx\\
&=\Pf(\uu_0,\lu_0, \r_0^{\a}) + \s_0(\sum_{k=1}^{\infty}\l_k - \L) -\frac{\a}{\a+1}\int \r_0^{\a + 1}dx\\
&=\sH_f(\uu_0, \lu_0).
\end{align*} 
\endprf


\begin{rem}
We note that the stationary states $(\uu_0,\lu_0,\muu_0, \r_0)$ for the quintic NLSS are shown to exist without necessary restriction on $\|\uu_0\|_{H^1_{\lu}(\tor^3)},$ but that we have only proven global existence of solutions to the quintic NLSS in the case $\|\uu_0\|_{H^1_{\lu}(\tor^3)} < \eta,$ for some $\eta > 0.$ It may be the case that there exists a choice of $f \in \sC$ that corresponds to a stationary state with large initial data, which would be an improvement to our results from Chapter 3. However, recall from the introduction that Ionescu and Pausader established the existence and uniqueness of global in time solutions to the defocusing quintic NLS on the square, rational 3-torus, for all $H^1$ initial data. It is our hope that future research will establish analogous results for the quintic NLS system, so that the nonlinear stability statement in  the quintic case of Theorem \ref{StabThm} will hold for all time.
\end{rem}

\subsection*{Acknowledgements}
We thank the anonymous referees for very helpful comments.
T.C. gratefully acknowledges support by the NSF through grants DMS-1151414 (CAREER), DMS-1716198, DMS-2009800, and the RTG Grant DMS-1840314 {\em Analysis of PDE}. A.U. was supported by NSF grants DMS-1716198 and DMS-2009800 through T.C., and by the RTG Grant DMS-1840314 {\em Analysis of PDE}.
\nocite{*}      
\bibliographystyle{plain}  
\bibliography{Chen-Urban-NLSS}        
\index{Bibliography@\emph{Bibliography}}

\end{document}